\documentclass[a4paper,UKenglish]{lipics-v2016}

\usepackage{microtype}
\usepackage{thmtools}
\usepackage{thm-restate}
\usepackage{xspace}
\usepackage{amsmath}
\usepackage{lineno}
\usepackage{xspace}
\usepackage{catchfilebetweentags}
\usepackage{algorithm}
\usepackage[noend]{algpseudocode}
\usepackage{wrapfig}
\graphicspath{{figures/}}
\usepackage{tikz}
\usetikzlibrary{arrows,intersections}

\theoremstyle{plain}

\title{Energy-efficient Delivery by Heterogeneous Mobile Agents\footnote{This work was partially supported by the project ANR-ANCOR (anr-14-CE36-0002-01), the SNF (project 200021L\_156620)
and the DFG Priority Programme 1736 “Algorithms for Big Data” (grant
SK 58/10-1).}}

\author[1]{Andreas Bärtschi}
\author[2]{Jérémie Chalopin}
\author[2]{Shantanu Das}
\author[3]{Yann Disser}
\author[1]{Daniel Graf}
\author[4]{Jan Hackfeld}
\author[1]{Paolo Penna}
\affil[1]{
	\small ETH Z\"{u}rich, Department of Computer Science, Switzerland\\
	\texttt{$\left\{\right.$baertschi,daniel.graf,paolo.penna$\left.\right\}$@inf.ethz.ch}
}
\affil[2]{
	\small LIF, CNRS and Aix-Marseille Université, France\\
  	\texttt{$\left\{\right.$shantanu.das,jeremie.chalopin$\left.\right\}$@lif.univ-mrs.fr}
}
\affil[3]{
	TU Darmstadt, 
  	\texttt{disser@mathematik.tu-darmstadt.de}
}
\affil[4]{
	TU Berlin, Institut für Mathematik, Technische Universität Berlin, Germany\\
  	\texttt{hackfeld@math.tu-berlin.de}
}

\authorrunning{Bärtschi, Chalopin, Das, Disser, Graf, Hackfeld and Penna}
\Copyright{A. Bärtschi, J. Chalopin, S. Das, Y. Disser, D. Graf, J. Hackfeld, P. Penna}



\ArticleNo{Page}

\newif\ifProofsAppendix
\ProofsAppendixfalse
\makeatletter
\def\CatchFBT@Fin@l#1[#2]{%
   \begingroup
      \makeatletter #2%
      \scantokens\expandafter{%
         \expandafter\CatchFBT@tok\expandafter{\the\CatchFBT@tok}}%
      \CatchFBT@IsAToken{#1}
         {\global#1\expandafter{\the\CatchFBT@tok}}
         {\xdef#1{\the\CatchFBT@tok}}%
      \ifx\CatchFBT@tok#1\else\global\CatchFBT@tok{}\fi
   \endgroup
}
\makeatother

\newcommand{\ourProblem}{\textsc{WeightedDelivery}\xspace}

\newcommand{\NP}{$\mathrm{NP}$}
\newcommand{\cost}{\textsc{cost}}
\newcommand{\ALG}{\textsc{Alg}}
\newcommand{\OPT}{\textsc{Opt}}
\newcommand{\SAT}{\textsc{Sol}}
\newcommand{\BoC}{\textsc{BoC}}
\providecommand{\dotcup}{\mathbin{\dot{\cup}}}

\newcommand{\drop}[1]{ }

\begin{document}

\maketitle

\begin{abstract}
	We consider the problem of delivering $m$ messages between specified source-target pairs in a weighted undirected graph, by $k$ mobile agents initially located at distinct nodes of the graph. 
	Each agent consumes energy proportional to the distance it travels in the graph and we are interested in optimizing the total energy consumption for the team of agents.
	Unlike previous related work, we consider heterogeneous agents with different rates of energy consumption (weights~$w_i$).
	To solve the delivery problem, agents face three major challenges: 
	\emph{Collaboration} (how to work together on each message), 
	\emph{Planning} (which route to take) and 
	\emph{Coordination} (how to assign agents to messages).

	We first show that the delivery problem can be 2-approximated \emph{without} collaborating and that this is best possible, i.e., we show that the \emph{benefit of collaboration} is 2 in general.
	We also show that the benefit of collaboration for a single message is~$1/\ln 2 \approx 1.44$.
	Planning turns out to be \NP-hard to approximate even for a single agent, but can be 2-approximated in polynomial time if agents have unit capacities and do not collaborate.
	We further show that coordination is \NP-hard even for agents with unit capacity, but can be efficiently solved exactly if they have uniform weights.	
	Finally, we give a polynomial-time $(4\max\tfrac{w_i}{w_j})$-approximation for message delivery with unit capacities.

\end{abstract}

\section{Introduction}
\label{sec:introduction}

Recent technological progress in robotics allows the mass production of inexpensive mobile robots which can be used to perform a variety of tasks autonomously without the need for human intervention. This gives rise to a variety of algorithmic problems for teams of autonomous robots, hereafter called \emph{mobile agents}. We consider here the delivery problem of moving some objects or messages between various locations. A mobile agent corresponds to an automated vehicle that can pick up a message at its source and deliver it to the intended destination. In doing so, the agent consumes energy proportional to the distance it travels. We are interested in energy-efficient operations by the team of agents such that the total energy consumed is minimized.

In general the agents may not be all identical; some may be more energy efficient than others if they use different technologies or different sources of power. We assume each agent has a given \emph{weight} which is the rate of energy consumption per unit distance traveled by this agent. Moreover, the agents may start from distinct locations. Thus it may be sometimes efficient for a agent to carry the message to some intermediate location and hand it over to another agent which carries it further towards the destination. On the other hand, an agent may carry several messages at the same time. Finding an optimal solution that minimizes the total energy cost involves scheduling the moves of the agents and the points where they pick up or handover the messages. 
We study this problem (called \ourProblem) for an edge-weighted graph $G$ which connects all sources and destinations. The objective is to deliver $m$ messages between specific source-target pairs using $k$ agents located in arbitrary nodes $G$.
Note that this problem is distinct from the connectivity problems on graphs or network flow problems since the initial  location of the agents are in general different from the sources where the messages are located, which means we need to consider the cost of moving the agents to the sources in addition to the cost of moving the messages. Furthermore, there is no one-to-one correspondence between the agents and the messages in our problem. 

Previous approaches to energy-efficient delivery of messages by agents have focused on a
bottleneck where the agents have limited energy (battery power) which restricts their movements~\cite{AnayaCCLPV16, DDalgosensors13}. The decision problem of whether a single message can be delivered without exceeding the available energy for any agent is known as the DataDelivery problem~\cite{DDicalp14} or the BudgetedDelivery problem~\cite{sirocco16} and it was shown to be weakly \NP-hard on paths~\cite{DDicalp14} and strongly \NP-hard on planar graphs~\cite{sirocco16}.

\subparagraph{Our Model.} 
We consider an undirected edge-weighted graph $G=(V,E)$. 
Each edge $e \in E$ has a \emph{cost} (or \emph{length}) denoted by $l_e$. 
The length of a simple path is the sum of the lengths of its edges. The distance between nodes $u$ and $v$ is denoted by $d_{G}(u,v)$ and is equal to the length of the shortest path from $u$ to $v$ in $G$. 
There are $k$ mobile agents denoted by $a_1, \dots a_k$ and having weights $w_1,\dots w_k$. These agents are initially located on arbitrary nodes $p_1, \ldots, p_k$
of $G$. We denote by $d(a_i,v)$ the distance from the initial location of~$a_i$ to node~$v$.
Each agent can move along the edges of the graph. Each time an agent $a_i$ traverses an edge $e$ it incurs an energy cost of $w_i \cdot l_e$.
Furthermore there are $m$ pairs of (source, target) nodes in $G$ such that for $1\leq i \leq m$, a message has to be delivered from  source node $s_i$ to a target node $t_i$. A message can be picked up by an agent from any node that it visits and it can be carried to any other node of $G$, and dropped there. 
The agents are given a \emph{capacity}~$\kappa$ which limits the number of messages an agent may carry simultaneously. 
There are no restrictions on how much an agent may travel. We denote by $d_j$ the total distance traveled by the $j$-th agent.
\ourProblem is the optimization problem of minimizing the total energy $\sum_{j=1}^k w_j d_j$ needed to deliver all messages.

A \emph{schedule} $S$ describes the actions of all agents as a sequence (ordered list) of pick-up actions $(a_j,p,m_i,+)$ and drop-off actions $(a_j,q,m_i,-)$, where each such tuple denotes the action of agent $a_j$ moving from its current location to node $p$ (node $q$) where it picks up message $m_i$ (drops message $m_i$, respectively). 
A schedule~$S$ implicitly encodes all the pick-up and drop-off times and it is easy to compute its total energy use of $\cost(S) := \sum_{j=1}^k w_j d_j$.
We denote by $S|_{a_j}$ the subsequence of all actions carried out by agent $a_j$ and by $S|_{m_i}$ the subsequence of all actions involving pick-ups or drop-offs of message $m_i$. 
We call a schedule \emph{feasible} if every pick-up action $(\text{\textunderscore},p,m_i,+),\ p\neq s_i$, is directly preceded by a drop-off action $(\text{\textunderscore},p,m_i,-)$ in $S|_{m_i}$ and if all the messages get delivered, see Figure~\ref{fig:schedule-example}.


\begin{figure}[t!]
	\centering
	\includegraphics[width=\textwidth]{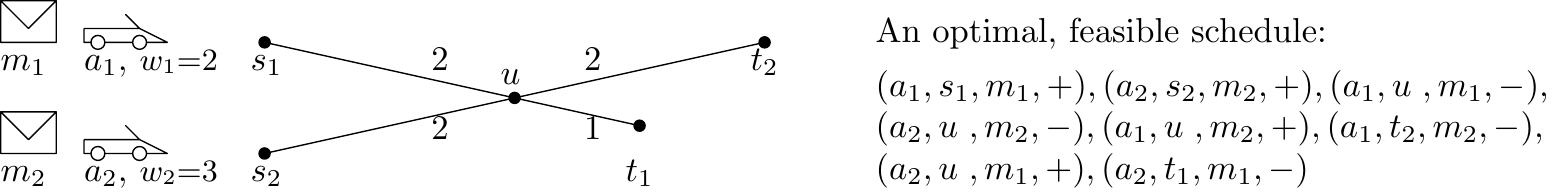}
	\caption{Example of an optimal, feasible schedule for two messages and two agents.}
	\label{fig:schedule-example}
\end{figure}

\subparagraph{Our Contribution.} 
Solving \ourProblem naturally involves simultaneously solving three subtasks, \emph{collaboration}, \emph{individual planning}, and \emph{coordination}:
First of all, if multiple agents work on the same message, they need to collaborate, i.e., we have to find all intermediate drop-off and pick-up locations of the message.
Secondly, if an agent works on more than one message, we have to plan in which order it wants to approach its subset of messages.
Finally, we have to coordinate which agent works on which subset of all messages (if they do this without collaboration, the subsets form a partition, otherwise the subsets are not necessarily pairwise disjoint).
Even though these three subtasks are interleaved, we investigate collaboration, planning and coordination separately in the next three sections. 
This leads us to a polynomial-time approximation algorithm for \ourProblem, given in Section~\ref{sec:approx}.

In Section~\ref{sec:singlemessage} we consider the \emph{Collaboration} aspect of \ourProblem. 
We first present a polynomial time solution for \ourProblem when there is only a single message ($m=1$). The algorithm has complexity  $O(|V|^3)$ irrespective of the number of agents $k$. 
In general, we show that any algorithm that only uses one agent for delivering every message cannot achieve an approximation ratio better than what we call the \emph{benefit of collaboration} $(\BoC)$ 
which is at least $1/ \ln \left( \left( 1 + 1/(2m)\right)^m \left( 1 + 1/(2m+1) \right) \right)$. We show this to be tight for $m=1$ (where $BoC \geq 1 / \ln 2$) and $m\rightarrow \infty$ (where $BoC \rightarrow 2$).

In Section~\ref{sec:planning} we look at the \emph{Planning} aspect of \ourProblem.
Individual planning by itself turns out to be \NP-hard on planar graphs and \NP-hard to approximate within a factor of less than $\tfrac{367}{366}$. On the positive side, we give approximation guarantees for restricted versions of \ourProblem which turn out to be useful for the analysis in Section~\ref{sec:approx}.

In Section~\ref{sec:message-order-hardness} we study the \emph{Coordination} aspect of \ourProblem. Even if collaboration and planning are taken care of (i.e., a schedule is fixed except for the assignment of agents to messages), 
Coordination also turns out to be \NP-hard even on planar graphs. 
The result holds for any capacity, including $\kappa=1$. This setting, however, becomes tractable if restricted to uniform weights of the agents.

In Section~\ref{sec:approx} we give a polynomial-time approximation algorithm for \ourProblem with an approximation ratio of $4\cdot \max \frac{w_i}{w_j}$ for $\kappa =1$. 
Due to the limited space, some proofs are deferred to the appendix.

\subparagraph{Related Work.} 
The problem of communicating or transporting goods between sources and destinations in a graph has been well studied in a variety of models with different optimization criteria. The problem of finding the smallest subgraph or tree that connects multiple sources and targets in a graph is called the \emph{point-to-point connection problem} and is known to be \NP-hard~\cite{Mccormick92Point-to-point}. The problem is related to the more well-known generalized Steiner tree problem~\cite{Winter87Steiner} which is also \NP-hard. 
Unlike these problems, the maximum flow problem in a network~\cite{edmonds1972theoretical}, puts a limit on the number of messages that can be transported over an edge, which makes the problem easier allowing for polynomial time solutions. In all these problems, however, there are no agents carrying the messages as in our problem.

For the case of a single agent moving in a graph, the task of optimally visiting all nodes, called the \emph{Traveling salesman problem} or visiting all edges, called the \emph{Chinese postman problem} have been studied before. The former is known to be \NP-hard~\cite{ApplegateTSP} while the latter can be solved in $O(|V|^{2}|E|)$ time~\cite{Edmonds73}. For metric graphs, the traveling salesman problem has a polynomial-time $\tfrac{3}{2}$-approximation for tours~\cite{Christofides76} and for paths with one fixed endpoint~\cite{Hoogeveen91}.
For multiple identical agents in a graph, 
Demaine et al.~\cite{Demaine2009} studied the problem of moving the agents to form desired configurations (e.g. connected or independent configurations) and they provided approximation algorithms and inapproximability results. Bilo et al.~\cite{Bilo2013} studied similar problems on visibility graphs of simple polygons and showed many motion planning problems to be hard to approximate.

Another optimization criteria is to minimize the maximum energy consumption by any agent, which requires partitioning the given task among the agents. 
Frederickson et al.~\cite{FredericksonHechtKim/76} studied this for uniform weights and called it the \emph{$k$-stacker-crane problem} and they gave approximation algorithms for a single agent and multiple agents.
Also in this minmax context, the problem of visiting all the nodes of a tree using $k$ agents starting from a single location is known to be \NP-hard~\cite{FraGKP04}. Anaya et al.~\cite{AnayaCCLPV16} studied the model of 
agents having limited energy budgets. They presented hardness results (on trees) and approximation algorithms (on arbitrary graphs) for the problem of transferring information from one agent to all others (\emph{Broadcast}) and from all agents to one agent (\emph{Convergecast}). For the same model, message delivery between a single $s$-$t$ node pair was studied by Chalopin et al.~\cite{DDalgosensors13, DDicalp14, sirocco16}\nocite{sirocco16arxiv}
as mentioned above. 
A recent paper~\cite{EnergyExchange15} shows that these three problems remain \NP-hard for general graphs even if the agents are allowed to exchange energy when they meet.

\section{Collaboration}
\label{sec:singlemessage}
In this section, we examine the \emph{collaboration} of agents: Given for each message $m_i$ all the agents $a_{i1}, a_{i2}, \ldots, a_{ix}$ which at some point carry the message,
one needs to find all pick-up and drop-off locations (handovers) $h_1, \ldots, h_y$ for the schedule entries
$(a_{i1}, \text{\textunderscore}, m_i,+), (a_{i1}, \text{\textunderscore}, m_i,-),$ $ \ldots, (a_{ix}, \text{\textunderscore}, m_i,-)$. 
Note, that in general we can have more than two action quadruples $(a_{ij}, \text{\textunderscore}, m_i,+/-)$ per agent $a_{ij}$.
When there is only a single message overall ($m=1$), we will use a structural result to tie together \ourProblem and Collaboration.
For multiple messages, however, this no longer holds: In this case, we analyze the benefit we lose if we forgo collaboration and deliver each message with a single agent.

\subsection{An Algorithm for WeightedDelivery of a Single Message}

\begin{restatable}{lem}{nonincreasing}
	\label{lemma:non-increasing}
	In any optimal solution to \ourProblem for a single message, if the message is delivered by agents with weights $w_1,w_2,\dots w_k$, in this order, then 
	(i) $w_{i} \geq w_{j}$ whenever $i<j$, and (ii) without loss of generality, $w_{i} \neq w_{j}$ for $i \neq j$.
	Hence there is an optimal schedule $S$ in which no agent $a_{j}$ has more than one pair of pick-up/drop-off actions. 
\end{restatable}

\ifProofsAppendix \else
	\ExecuteMetaData[appendix-collaboration.tex]{nonincreasing}
\fi

\drop{
\begin{corollary}
	There is an optimal schedule to \ourProblem for a single message $m_1$, in which no agent $a_{1j}$ has more than one pair of pick-up/drop-off actions in $S|_{a1j}$.
\end{corollary}
}

\begin{theorem}
An optimal solution of \ourProblem of a single message in a graph $G=(V,E)$ with $k \leq |V|$ agents can be found in $O(|V|^3)$ time. 
	\label{th:alg_single_msg}
\end{theorem}  

\begin{proof}
	We use the properties of Lemma~\ref{lemma:non-increasing} to create an auxiliary graph on which we run Dijkstra's algorithm for computing a shortest path from $s$ to $t$. 
	Given an original instance of single-message \ourProblem  consisting of the graph $G = (V,E)$, with $s,t \in V$, we obtain the auxiliary, \emph{directed} graph $G' = (V',E')$ as follows:
	\begin{itemize}
		\item	For each node $v \in V$ and each agent $a_i$, there is a node $v_{a_i}$ in $G'$.\\
			Furthermore $G'$ contains two additional vertices $s$ and $t$.
		\item	For $1\leq i \leq k$, there is an arc $(s,s_{a_i})$ of cost $w_i \cdot d_{G}(p_i,s)$ and an arc $(t_{a_i},t)$ of cost $0$. 
		\item	For $(u,v) \in E$ and $1\leq i\leq k$, there are two arcs $(u_{a_i},v_{a_i})$ and $(v_{a_i},u_{a_i})$ of cost $w_i \cdot l_{(u,v)}$.
		\item	For $u \in V$ and agents $a_i, a_j$ with $w_i>w_j$, there is an arc $(u_{a_i},u_{a_j})$ of cost $w_j \cdot d_{G}(p_j,u)$.
	\end{itemize}
	Note that any solution to the \ourProblem that satisfies the properties of Lemma~\ref{lemma:non-increasing} corresponds to some $s$-$t$-path in $G'$ 
	such that the cost of the solution is equal to the length of this path in $G_a$ and vice versa.
	This implies that the length of the shortest $s$-$t$ path in $G'$ is the cost of the optimal solution for \ourProblem in $G$. 
	Assuming that $k \leq |V|$, the graph $G'$ has $|V|\cdot k +2 \in O(|V|^2)$ vertices and at most $2k + (k^2|V| + |V|^2k)/2 -|V|\cdot k \in O(|V|^3)$ arcs. 
	The graph $G'$ can be constructed in $O(|V|^3)$ time if we use the \emph{Floyd Warshall} all pair shortest paths algorithm \cite{floyd1962algorithm, warshall1962theorem} in $G$.
	Finally, we compute the shortest path from $s$ to $t$ in $G'$ in time $O(|V|^3)$, using Dijkstra's algorithm with Fibonacci heaps.
\end{proof}

Unfortunately, the structural properties of Lemma~\ref{lemma:non-increasing} do not extend to multiple messages. 
In the next two subsections we investigate how the quality of an optimal solution changes if we only allow every message to be transported by one agent. 
Different messages may still be transported by different agents and one agent may also transport multiple messages at the same time as long as the number of messages is at most the capacity $\kappa$. 
To this end we define the \emph{Benefit of Collaboration} as the cost ratio between an optimal schedule $\OPT$ and a best-possible schedule without collaboration $S$, $\BoC = \min_S \cost(S)/\cost(\OPT)$.

\subsection{Lower Bound on the Benefit of Collaboration}


\begin{theorem}\label{theo-lb-boc}
On instances of \ourProblem with agent capacity $\kappa$ and $m$ messages, an algorithm using one agent for delivering every message cannot achieve an approximation ratio better than $1/ \ln \left( \left( 1 + 1/(2 r)\right)^r \left( 1 + 1/(2r+1) \right) \right)$, where $r:=\min\{\kappa,m\}$. 
\end{theorem}

\begin{figure}[b!]
\tikzstyle{graphnode}=[circle,draw,minimum size=3em,scale=0.53]
\tikzstyle{dnode}=[scale=0.5]
\newcommand*{\xFactor}{3}
\newcommand*{\yFactor}{1.2}
\begin{tikzpicture}
\foreach \j in {0,1} {
			\pgfmathsetmacro\koordX{-2+0.5*\j}
			\pgfmathsetmacro\koordY{0.6-0.15*\j}
			\node (a\j) at (\koordX*\xFactor,0.75*\yFactor) [graphnode] {$v_{1,\j}$};
			\node (b\j) at (\koordX*\xFactor,0.25*0.75*\yFactor) [graphnode] {$v_{2,\j}$};
			\node (c\j) at (\koordX*\xFactor,-0.75*\yFactor) [graphnode] {$v_{r,\j}$};
}
\node (a3) at (-0.7*\xFactor,0.75*\yFactor) [dnode] {$v_{1,n-1}$};
\node (b3) at (-0.7*\xFactor,0.25*0.75*\yFactor) [dnode] {$v_{2,n-1}$};
\node (c3) at (-0.7*\xFactor,-0.75*\yFactor) [dnode] {$v_{r,n-1}$};
\node (a3) at (-0.7*\xFactor,0.75*\yFactor) [graphnode] {};
\node (b3) at (-0.7*\xFactor,0.25*0.75*\yFactor) [graphnode] {};
\node (c3) at (-0.7*\xFactor,-0.75*\yFactor) [graphnode] {};

\node (d0) at(0*\xFactor,0) [dnode] {$v_{0,n}$};
\node (d1) at(0.75*\xFactor,0) [dnode] {$v_{0,n+1}$};
\node (d3) at (2*\xFactor,0) [dnode] {$v_{0,2n}$};
\node (d0) at(0*\xFactor,0) [graphnode] {};
\node (d1) at(0.75*\xFactor,0) [graphnode] {};
\node (d3) at (2*\xFactor,0) [graphnode] {};

\foreach \i in {a,b,c,d} {
		\foreach \j in {0} {
		\pgfmathtruncatemacro{\k}{\j+1}
			\draw (\i\j) -- (\i\k);	
		}
}
\draw (a3) -- (d0);	
\draw (b3) -- (d0);	
\draw (c3) -- (d0);	

\draw[dashed] (a1) -- (a3);	
\draw[dashed] (b1) -- (b3);	
\draw[dashed] (c1) -- (c3);	
\draw[dashed] (d1) -- (d3);	
\draw[loosely dotted](-2*\xFactor,-0.2*\yFactor)--(-2*\xFactor,-.4*\yFactor); 
\draw[loosely dotted](-1.5*\xFactor,-0.2*\yFactor)--(-1.5*\xFactor,-.4*\yFactor); 
\draw[loosely dotted](-0.7*\xFactor,-0.2*\yFactor)--(-0.7*\xFactor,-.4*\yFactor); 
\end{tikzpicture}
\caption{Lower bound construction for the benefit of collaboration.}\label{fig:boc-lower-bound}
\end{figure}

\begin{proof}
Consider the graph $G=(V,E)$ given in Figure~\ref{fig:boc-lower-bound}, where the length $l_e$ of every edge $e$ is $1/n$.
This means that $G$ is a star graph with center~$v_{0,n}$ and $r+1$ paths of total length 1 each. We have $r$ messages and message $i$ needs to be transported from $v_{i,0}$ to $v_{0,2n}$ for $i=1,\ldots, r$. There further is an agent $a_{i,j}$ with weight
$w_{i,j} = \tfrac{2r}{2r + j/n}$ starting at every vertex $v_{i,j}$ for $(i,j) \in \{1,\ldots, r\} \times \{0,\ldots, n-1\} \cup \{0\} \times \{n,\ldots, 2n\}$. 

We first show the following: If any agent transports $s$ messages $i_1,\ldots, i_s$ from $v_{i_j,0}$ to $v_{0,2n}$, then this costs at least $2s$. Note that this implies that any schedule $S$ for delivering all messages by the agents such that every message is only carried by one agent satisfies $\cost(S) \geq 2 r.$

So let an agent $a_{i,j}$ transport $s$ messages from the source to the destination $v_{0,2n}$. Without loss of generality let these messages be $1,\ldots, s$, which are picked up in this order. 
By construction, agent $a_{i,j}$ needs to travel a distance of at least $\tfrac{j}{n}$ to reach message~$1$, then distance 1 to to move back to $v_{0,n}$, then distance 2 for picking up message $i$ and going back to $v_{0,n}$ for $i=2,\ldots,s$, and finally it needs to move distance 1 from $v_{0,n}$ to $v_{0,2n}$. Overall, agent $a_{i,j}$ therefore travels a distance of at least $2s+\tfrac{j}{n}$.
The overall cost for agent $a_{i,j}$ to deliver the $s$ messages therefore is at least
\begin{linenomath}
\begin{align*}
(2s+\tfrac{j}{n}) \cdot w_{i,j} = (2s+\tfrac{j}{n})  \cdot \tfrac{2r}{2r + j/n} \geq
(2s+\tfrac{j}{n})  \cdot \tfrac{2s}{2s + j/n} = 2s.
\end{align*}
\end{linenomath}

Now, consider a schedule $S_\text{col}$, where the agents collaborate, i.e., agent $a_{i,j}$ transports message $i$ from $v_{i,j}$ to $v_{i,j+1}$ for $i=1,\ldots, r$, $j=0,\ldots,n-1$, where we identify $v_{i,n}$ with $v_{0,n}$. Then agent $a_{0,j}$ transports all $r$ messages from $v_{0,j}$ to $v_{0,j+1}$ for $j=n,\ldots, 2n-1$. This is possible because $r\leq \kappa$ by the choice of $r$. The total cost of this schedule is given by
\begin{linenomath}
\begin{align*}
\cost(S_\text{col})=r \cdot \int_0^1 f_\text{step} (x) dx +  \int_1^2 f_\text{step} (x) dx, 
\end{align*}
\end{linenomath}
where $f_\text{step} (x)$ is a step-function defined on $[0,2]$ giving the current cost of transporting the message, i.e., $f_\text{step} (x)=\tfrac{2 r}{2r+j/n}$ on the interval $[j/n,(j+1)/n)$ for $j=0,\ldots, 2n-1$. The first integral corresponds the the first part of the schedule, where the $r$ messages are transported separately and
therefore the cost of transporting message $i$ from $v_{i,j}$ to $v_{i,j+1}$ is exactly $\int_{j/n}^{(j+1)/n} f_\text{step} (x) dx = \tfrac{1}{n} \cdot  \tfrac{2 r}{2r+j/n}$. The second part of the schedule corresponds to the part, where all $r$ messages are transported together by one agent at a time.\\
Observe that the function $f(x)=2 r \cdot \tfrac{1}{2r-1/n+x}$ satisfies $f(x) \geq  f_\text{step} (x)$ on $[0,2]$, hence
\begin{linenomath}
\begin{align*}
\cost(S_\text{col}) & \leq r \int_0^1 f(x) dx +  \int_1^2 f (x) dx 
= 2r\left(r  \ln (2r - \tfrac{1}{n} +x)\Big|_0^1 +  \ln (2r - \tfrac{1}{n} +x)\Big|_1^2 \right) \\
&= 2 r  \ln \left( \left( \tfrac{2r - 1/n +1 }{2r-1/n}\right)^r \left( \tfrac{2r - 1/n +2 }{2r-1/n+1} \right) \right)
\stackrel{\smash{n\to\infty}}{\rightarrow}
2 r  \ln \left( \left(1 + \tfrac{1}{2r}\right)^r \left( 1 + \tfrac{1}{2r+1} \right) \right).
\end{align*}
\end{linenomath}
The best approximation ratio of an algorithm that transports every message by only one agent compared to an algorithm that uses an arbitrary number of agents for every message is therefore bounded from below by
\begin{linenomath}
\begin{align*}
BoC \geq \min_S \cost(S)/ \cost(S_\text{col}) \geq 1 / \ln \left( \left(1 + \tfrac{1}{2r}\right)^r \left( 1 + \tfrac{1}{2r+1} \right) \right).	& && \qedhere 
\end{align*}
\end{linenomath}
\end{proof}

By observing that $\lim_{r\to\infty} 1/ \ln \left( \left( 1 + 1/(2 r)\right)^r \left( 1 + 1/(2r+1) \right) \right) =1/ \ln \left( e^{1/2} \right) = 2$, we obtain the following corollary.


\begin{corollary}\label{cor-lb-boc}
A schedule for \ourProblem where every message is delivered by a single agent cannot achieve an approximation ratio better than $2$ in general, and better than~$1/\ln 2 \approx 1.44$ for a single message.
\end{corollary}

\subsection{Upper Bounds on the Benefit of Collaboration}

We now give tight upper bounds for Corollary~\ref{cor-lb-boc}.
The following theorem shows that the benefit of collaboration is~2 in general. 
We remark that finding an optimal schedule in which every message is transported from its source to its destination by one agent, is already \NP-hard, as shown in Theorem~\ref{thm:single-agent-hardness}.
 
\begin{restatable}{thm2}{BoCUB}
	\label{theo-ub-boc}
	Let $\OPT$ be an optimal schedule for a given instance of \ourProblem. Then there exists a schedule $S$ such that every message is only transported by one agent and $\cost(S)\leq 2 \cdot \cost(\OPT)$. 
\end{restatable}

\ifProofsAppendix 
\begin{proof}[Proof Sketch]
We may assume (without loss of generality) that the optimal schedule~\OPT{} transports each message along a simple path, and we construct a directed multigraph~$G_S$ with the same set of vertices and an arc for every move of an agent in~\OPT{}.
We label every arc by the exact set of messages that were carried during the corresponding move in~\OPT{}.
For every arc in~$G_S$ we add a backwards arc with the same label.

Obviously, every connected component of~$G_S$ is Eulerian, and we claim that any agent in each component can follow some Eulerian tour that allows to deliver all messages.
In particular, the agent needs exactly twice as many moves as the total number of moves of all agents in the component in~\OPT{}.
If we choose the cheapest agent (in terms of weight) in each component, we obtain a tour with at most twice the cost of~\OPT{}.

We compute the Eulerian for the cheapest agent of a component as a combination of multiple tours, respecting arc labels in the following sense:
During every move along a forward arc, the agent carries the exact set of messages prescribed by the arc label, and during every move along a backward arc, the agent does not carry any messages.
This ensures that all messages travel along the same path as in~\OPT{}.
Whenever the agent is at a vertex~$v$ and is missing message~$i$ in order to proceed along some path, this means the current vertex must lie on $i$'s path in~\OPT{}, and thus there must be a path of backwards edges to the current location of~$i$.
The agent follows this path and recursively brings the message back to~$v$.
In the process, more recursive calls may be necessary, but we can prove that there cannot be a circular dependence between messages.

Therefore, the procedure eventually terminates after computing a closed tour.
Note that, so far, the tour is still ``virtual'' in the sense that the agent didn't actually move but merely computed the tour.
We remove the tour from~$G_S$, update all message positions, and recursively apply the procedure starting from the last vertex along the tour that is still adjacent to untraversed edges.
By combining all (virtual) tours that we obtain in the recursion, we eventually get a Eulerian tour for the agent that obeys all arc labels.
This means that the agent can successfully simulate all moves in~\OPT{} while ensuring that it is carrying the required messages before each move.
\end{proof}
\else 
	\ExecuteMetaData[appendix-collaboration.tex]{BoCUB} 
\fi

\subparagraph{Single Message.} For the case of a single message, we can improve the upper bound of~2 on the benefit of collaboration from Theorem~\ref{theo-ub-boc}, to a tight bound of $1/\ln 2 \approx 1.44$.
\ifProofsAppendix  
 The idea of the proof is to use that the weights are non-increasing by Lemma~\ref{lemma:non-increasing}. After scaling appropriately, we assume the message path to be the interval $[0,1]$ and then choose a $b$ such that the function $\tfrac{b}{x+1}$ is a lower bound on the weight of the agent transporting the message at point $x$ on the message path. The intersection of $\tfrac{b}{x+1}$ and the step-function $f$ representing the weight of the agent currently transporting the message then gives an agent, which can transport the message with at most $(1/ \ln 2)$-times the cost of an optimal schedule. 
\fi

\begin{restatable}{thm2}{BoConemessageUB}
	\label{th:BoC}
	There is a $(1/ \ln 2)$-approximation algorithm using a single agent for \ourProblem with~$m=1$.
\end{restatable}

\ifProofsAppendix \else
\ExecuteMetaData[appendix-collaboration.tex]{BoConemessageUB}
\fi

\subparagraph{No Intermediate Dropoffs.} In the following we show that for $\kappa \in \left\{ 1, \infty \right\}$ the upper bound of~$2$ on the benefit of collaboration still holds, 
with the additional property that each message is carried by its single agent without any intermediate dropoffs. 
We will make use of this result later in the approximation algorithm for \ourProblem with $\kappa=1$ (Section~\ref{sec:approx}).

\begin{restatable}{thm2}{nointermediate}
	\label{theo-boc-nointermediate}
	Let $\OPT$ be an optimal schedule for a given instance of \ourProblem with $\kappa \in \left\{ 1,\infty \right\}$. 
	Then there exists a schedule $S$ such that (i) every message is only transported by a single agent, with exactly one pick-up and one drop-off, 
	(ii) $\cost(S)\leq 2 \cdot \cost(\OPT)$, and (iii) every agent $a_j$ returns to its starting location $p_j$.
\end{restatable}

\ifProofsAppendix \else
\ExecuteMetaData[appendix-collaboration.tex]{nointermediate}
\fi

\section{Planning}
\label{sec:planning}

We now look in isolation at the problem of ordering the messages within the schedule of an agent, which we call \emph{Planning}. Formally, the \emph{Planning} aspect of \ourProblem is the following task: Given a schedule $S$ and one of its agents $a_j$, reorder the actions in $S|_{a_j}$ in such a way that the schedule remains feasible and the costs are minimized.

Generally speaking, for a complex schedule with many message handovers, the reordering options for a single agent $a_j$ might be very limited. 
First of all, we must respect the capacity of $a_j$, i.e., in every prefix of $S|_{a_{j}}$, the number of pick-up actions $(a_j,*,*,+)$ cannot exceed the number of drop-off actions $(a_j,*,*,-)$ by more than $\kappa$.
Even then, reordering $S|_{a_j}$ might render $S$ infeasible because of conflicts with some other subschedule $S|_{a_x}$. 
But \emph{Planning} also includes the instances where a single agent delivers all the messages, one after the other straight to the target, and where the only thing that has to be decided is the ordering. We show now that in this setting, where there is no non-trivial \emph{coordination} or \emph{collaboration} aspect, \ourProblem is already \NP-hard.




\begin{restatable}{thm2}{singleagenthardness}
\label{thm:singleagenthardness}
	\emph{Planning} of \ourProblem problem is \NP-hard for all capacities $\kappa$ even for a single agent on a planar graph.
	\label{thm:single-agent-hardness}
\end{restatable}

\ifProofsAppendix 
The hardness follows by a reduction from Hamiltonian cycles on a grid graph $H$, a problem shown to be \NP-hard by Itai et al.~\cite{NPonGrid}:
We put an isolated message at every node of $H$, forcing the agent to visit each node exactly once. A longer isolated edge at the start forces the agent to come back to the start at the end, see Figure~\ref{fig:hamilton-cycle-reduction}.
\begin{figure}[t!]
	\centering
	\includegraphics[width=\linewidth]{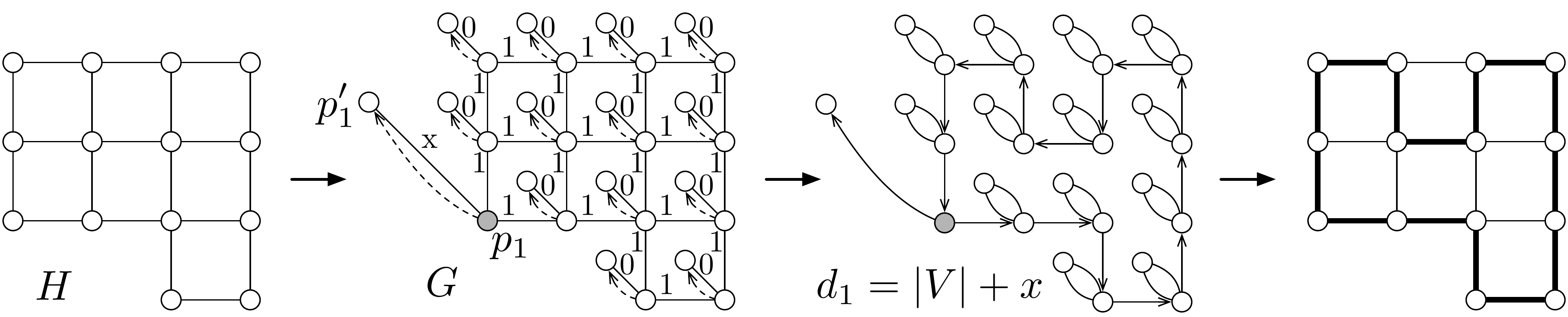}
	\caption{Finding a Hamiltonian cycle via \ourProblem with a single agent. Picking $x$ to be large enough, e.g. $x = \lvert V \rvert$, allows us to enforce that the agent will end in $p_1'$.}
	\label{fig:hamilton-cycle-reduction}
\end{figure}
\else
\ExecuteMetaData[appendix-planning.tex]{singleagenthardness}
\fi

Using similar ideas, we can use recent results for the approximation hardness of metric TSP~\cite{karpinski2015new} to immediately show that \emph{Planning} of \ourProblem 
can not be approximated arbitrarily well, unless $\mathrm{P} = \mathrm{NP}$.

\begin{restatable}{thm2}{approximationhardness}
	It is \NP-hard to approximate the \emph{Planning} of \ourProblem to within any constant approximation ratio less than 367/366.
\end{restatable}

\ifProofsAppendix \else
\ExecuteMetaData[appendix-planning.tex]{approximationhardness}
\fi

\subsection{Polynomial-time Approximation for Planning in Restricted Settings}

Motivated by Theorem~\ref{theo-boc-nointermediate}, we now look at the restricted setting of planning for a feasible schedule $S^R$ 
of which we know that each message is completely transported by some agent $a_j$ without intermediate drop-offs, 
i.e., for every message~$m_i$ there must be an agent~$j$ with $S^R|_{m_i} = (a_j, s_i, m_i, +), (a_j,t_i, m_i, -)$.
This allows us to give polynomial-time approximations for planning with capacity $\kappa \in \left\{ 1,\infty \right\}$:

\begin{restatable}{thm2}{restrictedplanning}
	\label{thm:planning-restricted}
	Let $S^R$ be a feasible schedule for a given instance of \ourProblem with the restriction that $\forall i\ \exists j:\ S^R|_{m_i} = (a_j, s_i, m_i, +), (a_j,t_i, m_i, -)$.
	Denote by $\OPT(S^R)$ a reordering of $S^R$ with optimal cost.
	There is a polynomial-time planning algorithm $\ALG$ which gives a reordering $\ALG(S^R)$ such that
	$\cost(\ALG(S^R)) \leq 2\cdot \cost(\OPT(S^R))$ if $\kappa=1$ and	
	$\cost(\ALG(S^R)) \leq 3.5\cdot \cost(\OPT(S^R))$ if $\kappa=\infty$.
\end{restatable}

\begin{proof} 
	By the given restriction, separate planning of each $S^R|_{a_j}$ independently maintains feasibility of $S^R$. We denote by $m_{j1}, m_{j2}, \ldots, m_{jx}$ the messages appearing in $S^R_{a_j}$. 
	We define a complete undirected auxiliary graph $G'=(V',E')$ on the node set 
	$V' = \left\{ p_j \right\} \cup \left\{ s_{j1}, s_{j2}, \ldots, s_{jx} \right\} \cup \left\{ t_{j1}, \ldots, t_{jx} \right\}$ with edges $(u,v)$ having weight $d_G(u,v)$.

	For $\underline{\kappa = 1}$, the schedule $\OPT(S^R)|_{a_j}$ corresponds to a Hamiltonian path $H$ in $G'$ of minimum length, starting in $p_j$, 
	subject to the condition that for each message $m_{ji}$ the visit of its source $s_{ji}$ is directly followed by a visit of its destination $t_{ji}$.
	We can lower bound the length of $H$ with the total length of a spanning tree $T'=(V', E(T)') \subseteq G'$ as follows: 
	Starting with an empty graph on $V'$ we first add all edges $(s_{ji},t_{ji})$. 
	Following the idea of Kruskal~\cite{Kruskal56}, we add edges from 
	$\left\{ p_j \right\} \times \left\{ s_{j1}, \ldots, s_{jx} \right\} \cup \left\{ t_{j1}, \ldots, t_{jx} \right\} \times \left\{ s_{j1}, \ldots, s_{jx} \right\}$ 
	in increasing order of their lengths, disregarding any edges which would result in the creation of a cycle.
	Now a DFS-traversal of $T'$ starting from $p_j$ visits any edge $(s_{ji},t_{ji})$ in both directions. Whenever we cross such an edge from $s_{ji}$ to $t_{ji}$, 
	we add $(a_j,s_{ji},m_{ji},+),(a_j,t_{ji},m_{ji},-)$ as a suffix to the current schedule $\ALG(S^R)|_{a_j}$. 
	We get an overall cost of $\cost(\ALG(S^R)|_{a_j}) \leq 2\cdot \sum_{e \in E(T')} l_e \leq 2 \cdot \sum_{e \in H} = 2\cdot \cost(\OPT(S^R)|_{a_j})$.
\ifProofsAppendix 
	For $\underline{\kappa = \infty}$, the idea is to first collect all messages by traversing a spanning tree (with cost $\leq 2\cdot \cost(\OPT(S^R)|_{a_j})$) 
	and then delivering all of them in a metric TSP path fashion (with cost $\leq \tfrac{3}{2}\cdot \cost(\OPT(S^R)|_{a_j})$), see Appendix~\ref{app:planning} for details. 
\else
\ExecuteMetaData[appendix-planning.tex]{restrictedplanning}
\fi 
\end{proof}

\begin{remark}
	If we assume as an additional property that the agent returns to its starting position $p_j$ (as for example in the result of Theorem~\ref{theo-boc-nointermediate}), we can get a better approximation for the case $\kappa=1$.
	Instead of traversing a spanning tree twice, we can model this as the \emph{stacker-crane problem} for which a polynomial-time $1.8$-approximation is known~\cite{FraGKP04}.
\end{remark}

\section{Coordination}
\label{sec:message-order-hardness}
In this section, we restrict ourselves to the \emph{Coordination} aspect of \ourProblem. 
We assume that collaboration and planning are taken care of. 
More precisely, we are given a sequence containing the complete \emph{fixed} schedule $S^-$ of all actions
$(\text{\textunderscore}, s_i, m_i, +),\dots, (\text{\textunderscore}, h, m_i,-),$ $\ldots,(\text{\textunderscore}, t_j, m_j, -),$ but without an assignment of the agents to the actions. 
Coordination is the task of assigning agents to the given actions. 
Even though coordination appears to have the flavor of a matching problem, it turns out to be \NP-hard to optimally match up agents with the given actions.
This holds for any capacity, in particular for $\kappa=1$. The latter, however, has a polynomial-time solution if all agents have uniform weight.

\subsection{NP-Hardness for Planar Graphs}

We give a reduction from planar 3SAT: From a given planar 3SAT formula $F$ we construct an instance of \ourProblem 
that allows a schedule $S$ with ``good'' energy \cost($S$) if and only if $F$ is satisfiable.


\subparagraph{Planar 3SAT.} Let $F$ be a three-conjunctive normal form (3CNF) with $x$ boolean variables $V(F)=\left\{ v_1, \dots, v_x \right\}$ and $y$
clauses $C(F)=\left\{ c_1, \dots, c_y \right\}$. Each clause is given by a subset of at most three literals of the form $l(v_i) \in \left\{ v_i,
\overline{v_i} \right\}$. We define a corresponding graph $H(F) = (N,A)$ with a node set consisting of all clauses and all variables ($N = V(F) \cup C(F)$).
We add an edge between a clause $c$ and a variable $v$, if $v$ or $\overline{v}$ is contained in $c$. Furthermore we add a cycle consisting of edges
between all pairs of consecutive variables, i.e., $A = A_1 \cup A_2$, where
$\ A_1 = \left\{ \left\{ c_i, v_j \right\} \ | \ v_j \in c_i \text{ or } \overline{v_j} \in c_i \right\}, 
 \ A_2 = \left\{ \left\{ v_j, v_{(j \bmod x)+1} \right\} \ | \ 1\leq j \leq x \right\}.$
We call $F$ \emph{planar} if there is a plane embedding of $H(F)$. The \emph{planar 3SAT} problem of deciding whether a given planar 3CNF $F$ is
satisfiable is known to be \NP-complete. Furthermore the problem remains \NP-complete \emph{if at each variable node} the plane
embedding is required to have all arcs representing positive literals on one side of the cycle $A_2$ and all arcs representing negative literals on
the other side of $A_2$~\cite{planar3sat82}. We will use this \emph{restricted version} in our reduction and assume without loss of generality that
the graph $H(F)\setminus A_2$ is connected and that $H(F)$ is a simple graph (i.e. each variable appears at most once in every clause).

\subparagraph{Building the Delivery Graph.} We first describe a way to transform any planar 3CNF graph $H(F)$ into a planar delivery graph $G
= G(F)$, see Figure~\ref{fig:planar3sat}.
%
\begin{figure}[t!]
	\centering
	\includegraphics[width=\linewidth]{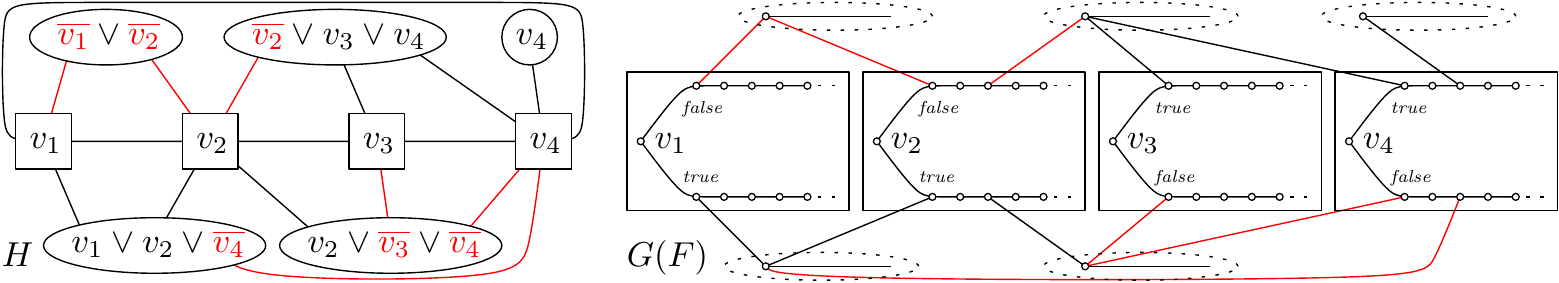}
	\caption{(left) A restricted plane embedding of a 3CNF $F$ which is satisfied by $(v_1,v_2,v_3,v_4) = (true, false, false, true)$. 
	(right) Its transformation to the corresponding delivery graph.}
	\label{fig:planar3sat}
\end{figure}
We transform the graph in five steps: First we delete all edges of the cycle $A_2$, but we keep in mind that at each variable node all positive
literal edges lie on one side and all negative literal edges on the other side. Secondly let $\deg_{H(F),A_1}(v)$ denote the remaining degree of a
variable node $v$ in $H$ and surround each variable node by a \emph{variable box}. A variable box contains two paths adjacent to $v$ on which
internally we place $\deg_{H(F),A_1}(v)$ copies of $v$: One path (called henceforth the \emph{$\mathit{true}$-path}) contains all nodes having an
adjacent positive literal edge, the other path (the \emph{$\mathit{false}$-path}) contains all nodes having an adjacent negative literal edge. In a
next step, we add a single node between any pair of node copies of the previous step. As a fourth step, we want all paths to contain the same number
of nodes, hence we fill in nodes at the end of each path such that every path contains exactly $2y \geq 2\deg_{H(F),A_1}(v)$ internal nodes. Thus each
variable box contains a variable node $v$, an adjacent $\mathit{true}$-path (with internal nodes $v_{\mathit{true},1}, \ldots, v_{\mathit{true},2y-1}$
and a final node $v_{\mathit{true},2y}$) and a respective $\mathit{false}$-path.
Finally for each clause node $c$ we add a new node $c'$ which we connect to $c$. 
The new graph $G(F)$ has polynomial size and all the steps can be implemented in such a way that $G(F)$ is planar.

\subparagraph{Messages, Agents and Weights.} We are going to place one \emph{clause message} on each of the $y$ clause nodes and a \emph{literal
message} on each of the $2x$ paths in the variable boxes for a total of $4xy$ messages. More precisely, on each original clause node $c$ we place
exactly one clause message which has to be delivered to the newly created node $c'$. Furthermore we place a literal message on every internal node
$v_{\mathit{true},i}$ of a $\mathit{true}$-path and set its target to $v_{\mathit{true},i+1}$ (same for the $\mathit{false}$-path).
We set the length of all edges connecting a source to its target to 1.

Next we describe the locations of the agents in each variable box. We place one \emph{variable agent} of weight 1 on the variable node $v$.  The
length of the two adjacent edges are set to $\varepsilon$, where $\varepsilon:= (8xy)^{-2}$. Furthermore we place $y$ \emph{literal agents} on each
path: The $i$-th agent is placed on $v_{\mathit{true},2(y-i)}$ (respectively $v_{\mathit{false},2(y-i)}$) and gets weight $1+i\varepsilon$.  It
remains to set the length of edges between clause nodes and internal nodes of a path. By construction the latter is the starting position of an agent
of uniquely defined weight $1+i\varepsilon$; we set the length of the edge to $\tfrac{1-i\varepsilon}{1+i\varepsilon}$. For an illustration see
Figure~\ref{fig:weight-hardness}, where each agent's starting location is depicted by a square and each message is depicted by a colored arrow.

\begin{figure}[t!]
	\centering
	\includegraphics[width=\linewidth]{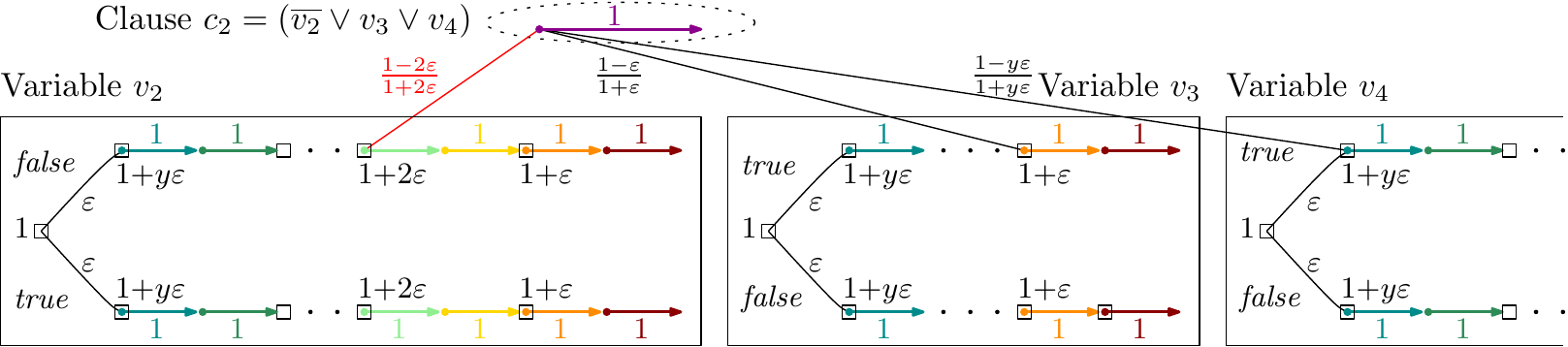}
	\caption{Agent positions ($\square$) and weights (in black); Messages ($\rightarrow$) and edge lengths (in color).}
	\label{fig:weight-hardness}
\end{figure}

\subparagraph{Reduction.} The key idea of the reduction is that for each variable $u$, the corresponding variable box contains a variable agent who can
\emph{either} deliver all messages on the $\mathit{true}$-path (thus setting the variable to true), \emph{or} deliver all messages on the
$\mathit{false}$-path (thus setting the variable to false). Assume $u$ is set to true. If $u$ is contained in a clause $c$, then on the adjacent node 
$v_{\mathit{true},i}$ there is a (not yet used) literal agent. Intuitively, this agent was \emph{freed by the variable agent} and can thus be sent to 
deliver the clause-message. If $\overline{u}$ is contained in $c$, the corresponding literal on the $\mathit{false}$-path can't be sent to deliver the
clause message, since it needs to transport messages along the $\mathit{false}$-path.

\ifProofsAppendix 
There is such a feasible schedule $\SAT$ of the agents in $G(F)$ if and only if there is a satisfiable assignment (a solution) for the variables of a 3CNF $F$.
Its total (energy) cost is $\cost(\SAT) := 4xy + 2y + x(y^2+y+1)\varepsilon$ (Appendix~\ref{app:coordination}: Lemma~\ref{lem:optimum-cost}).
Furthermore, we can show that \emph{any} schedule $S$ which doesn't correspond to a satisfiable variable assignment has cost $\cost(S) > \cost(\SAT)$ 
(Appendix~\ref{app:coordination}: Lemmata~\ref{lem:variablebox-independence},~\ref{lem:agent-movement} and~\ref{lem:optimum-schedule-cost}).
This is true independent of whether $S$ adheres to the schedule without agents $S^-$ or not and holds for any capacity $\kappa$.


\subparagraph{Fixed Sequence (Schedule without Agent Assignment).} It remains to fix a sequence $S^-$ that describes the schedule $\SAT$ described in the reduction idea but which does not allow us to infer a satisfiable assignment:
This is the case for any $S^-$ consisting of consecutive pairs $(\text{\textunderscore}, s_i, m_i, +), (\text{\textunderscore}, t_i, m_i, -)$
such that if $m_i$ lies to the left of $m_j$ on some $\mathit{true}$- or $\mathit{false}$-path, it precedes $m_j$ in the schedule.

\else
\ExecuteMetaData[appendix-coordination.tex]{messageorderhardnessthm1}
\ExecuteMetaData[appendix-coordination.tex]{messageorderhardnessthm2}
\fi

\begin{restatable}{thm2}{messageorderhardnessthm}
	\emph{Coordination} of \ourProblem is \NP-hard on planar graphs for all capacities $\kappa$, even if we are given prescribed collaboration and planning.
	\label{thm:message-order-hardness}
\end{restatable}

\subsection{Polynomial-time Algorithm for Uniform Weights and Unit Capacity}
\label{sec:mcmf}

Note that Coordination is \NP-hard even for capacity $\kappa=1$. Next we show that this setting is approachable once we restrict ourselves to uniform weights. 


\begin{theorem}
	Given collaboration and planning in the form of a complete schedule with missing agent assignment,
	Coordination of \ourProblem with capacity $\kappa=1$ and agents having uniform weights can be solved in polynomial time.
\label{thm:uniform-matching}
\end{theorem}

\begin{proof}
	As before, denote by $S^- = (\text{\textunderscore}, s_i, m_i, +),\dots, (\text{\textunderscore}, h, m_i,-),$ $\ldots,(\text{\textunderscore}, t_j, m_j, -)$ the prescribed schedule without agent assignments.
	Since all agents have the same uniform weight $w$, the cost $\cost(S)$ of any feasible schedule $S$ is determined by $\cost(S) = w\cdot \sum_{j=1}^k d_j$. 
	Hence at a pick-up action $(\text{\textunderscore}, q, m_i, +)$ it is not so much important \emph{which} agent picks up the message as \emph{where / how far} it comes from.
	
	Because we have capacity $\kappa = 1$, we know that the agent has to come from either its starting position or from a preceding drop-off action $(\text{\textunderscore}, p, m_j, -) \in S^-$.
	This allows us to model the problem as a weighted bipartite matching, see Figure~\ref{fig:MCMF-example} (center). We build an auxiliary graph $G' = (A \cup B, E_1' \cup E_2')$. 
	A maximum matching in this bipartite graph will tell us for every pick-up action in $B$, where the agent that performs the pick-up action comes from in $A$.
	Let $A := \{p_1, \dots, p_k\} \cup \{ (\text{\textunderscore},*,*,-)	\}$ and $B := \{ (\text{\textunderscore},*,*,+)	\}$.
	We add edges between all agent starting positions and all pick-ups, $E_1' :=  \{p_1, \dots, p_k\} \times \left\{ (\text{\textunderscore},q,m,+) \ \mid \ (\text{\textunderscore},q,m,+) \in B \right\}$
	of weight $d_G(p_i, q)$. Furthermore we also add edges between drop-offs and all subsequent pick-ups $E_2' := \left\{ \left( (\text{\textunderscore}, p, m_j, -), (\text{\textunderscore}, q, m_i, +)  \right) \ \mid \ 
	(\text{\textunderscore}, p, m_j, -) < (\text{\textunderscore}, q, m_i, +) \text{ in }S^- \right\}$
	of weight $d_G(p,q)$.

	\begin{figure}[t!]
		\centering
		\includegraphics[width=0.95\linewidth]{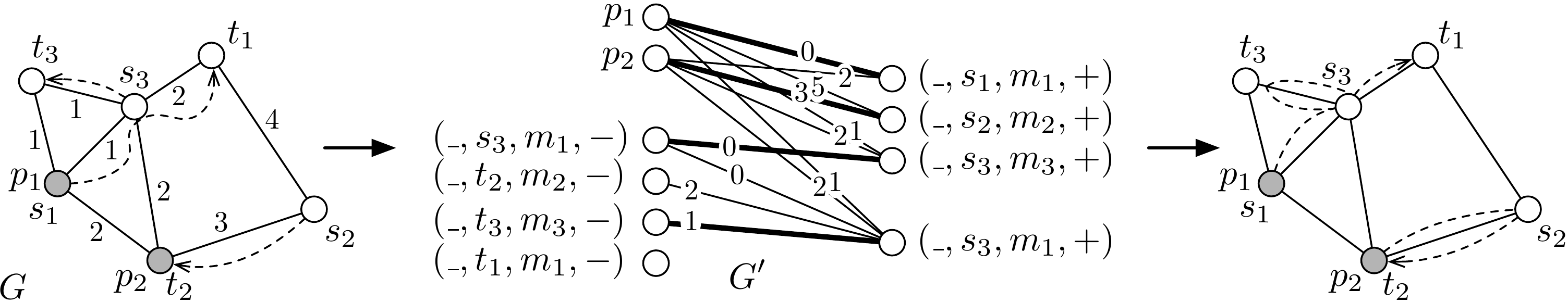}
		\caption{Illustration of the coordination of the following schedule $S = (\text{\textunderscore}, s_1, m_1, +)$, $(\text{\textunderscore}, s_2, m_2, +), (\text{\textunderscore}, s_3, m_1, -), (\text{\textunderscore}, s_3, m_3, +), (\text{\textunderscore}, t_2, m_2, -), (\text{\textunderscore}, t_3, m_3, -), (\text{\textunderscore}, s_3, m_1, +), (\text{\textunderscore}, t_1, m_1, -)$ (left) Instance with 3 messages and 2 agents of uniform weight. (center) Equivalent weighted bipartite matching problem $G'$. (right) The resulting trajectories of the agents.}
		\label{fig:MCMF-example}
	\end{figure}

	A maximum matching of minimum cost in $G'$ captures the optimal assignment of agents to messages and can be found by solving the classic \emph{assignment problem}, 
	a special case of the \emph{minimum cost maximum flow problem}. 
	Both of these problems can be solved in polynomial time for instance using the \emph{Hungarian method}~\cite{kuhn1955hungarian} or the \emph{successive shortest path algorithm}~\cite{edmonds1972theoretical}, respectively.
	The cost of this optimum matching corresponds to the cost of the agents moving around without messages. 
	The cost of the agents while carrying the messages can easily be added: Consider the schedule $S^-$ restricted to a message $m_i$. 
	This subsequence $S^-|_{m_i}$ is a sequence of pairs of pick-up/drop-off actions $\left( (\text{\textunderscore}, q, m_i, +), (\text{\textunderscore}, p, m_i, -)  \right)$, 
	and in every pair the message is brought from $q$ to $p$ on the shortest path, so we add $\sum d_G(q,p)$. 
	Concatenating these piecewise shortest paths gives the trajectory of each agent in the optimum solution, as illustrated in Figure~\ref{fig:MCMF-example}~(right).
\end{proof}

Our algorithm is remotely inspired by a simpler problem at the ACM ICPC world finals 2015~\cite{ICPC}.
The official solution is pseudo-polynomial~\cite{ICPCofficial}, Austrin and Wojtaszczyk~\cite{ICPCunofficial} later sketched a min-cost bipartite matching solution.

\section{Approximation Algorithm}
\label{sec:approx}

\ifProofsAppendix 
We have seen in Section~\ref{sec:singlemessage} that the cost of an optimal schedule $\OPT$ for \ourProblem can be approximated to within a factor of two by restricting ourselves to schedules $S^R$ in which every message is only transported by a single agent, with exactly one pick-up and one drop-off (Theorem~\ref{theo-boc-nointermediate}).
Let $\OPT^R$ be an optimal restricted schedule. Given an auxiliary graph $G'$ on a vertex set consisting of all agent positions and all message source and destination nodes, 
we construct in polynomial time a minimum tree cover of $G'$ from which be build a schedule $S^*$ in which agents travel at most twice the distance of their counterparts in $\OPT^R$ (similarly to Theorem~\ref{thm:planning-restricted}).
Here we require capacity $\kappa=1$. 
In order for our method to work, we need to be indifferent between different weights; 
we achieve this by boosting each agent's weight $w_j$ to $\max w_i$, resulting in an additional loss in the approximation factor of $\max \tfrac{w_i}{w_j}$: 
\else
By putting the results of the previous sections together, we obtain the following approximation algorithm.
\fi

\begin{restatable}{thm2}{approximationalgorithm}
	There is a polynomial-time $( 4\max\tfrac{w_i}{w_j} )$-approximation algorithm for \newline \ourProblem with capacity $\kappa=1$.
\end{restatable}

\ifProofsAppendix \else
\ExecuteMetaData[appendix-approximation.tex]{approximationalgorithm}
\fi

\newpage
\bibliography{bib}

\ifProofsAppendix 
\newpage
\section*{Appendix}
\appendix

\section{Appendix: Collaboration}
\label{app:BoC}

\subsection*{An algorithm for WeightedDelivery of a single message}

\nonincreasing*

\begin{proof}
	If agent $a_i$ hands the message over to agent $a_j$ with $w_{i}<w_{j}$ in any solution, we can construct a better solution by replacing
	$a_j$'s trajectory carrying the message with the same trajectory using agent $a_i$. By the same argument we may also assume without loss of
	generality that the weights of the agents carrying the message in an optimum schedule are strictly decreasing, since we can merge trajectories
	of equal weight.
\end{proof}

\begin{example}
This example shows that Lemma~\ref{lemma:non-increasing} is not true for more than one message.
In the graph shown in Figure~\ref{fig:example_increasing_weights}, we let one agent $a_1$ with weight $w_1=1$ start in vertex $s_2$ and a second agent $a_2$ with weight $w_2=1.5$ start in vertex $v_1$. In the optimal schedule, the message 2 with starting location $s_2$ is first transported by $a_1$ to $v_1$ and from there by $a_2$ to its destination $t_2$. Thus, the weights of the agents transporting the message $2$ are increasing in this case.
\begin{figure}[h!]
\tikzstyle{graphnode}=[circle,draw,minimum size=2.5em,scale=0.7]
\tikzstyle{dnode}=[scale=0.7]
\newcommand*{\xFactor}{1.7}
\newcommand*{\yFactor}{1.2}
\begin{tikzpicture}
\node (s2) at (0*\xFactor,0*\yFactor) [graphnode] {$s_{2}$};
\node (v1) at (1*\xFactor,0*\yFactor) [graphnode] {$v_{1}$};
\node (t2) at (2*\xFactor,0*\yFactor) [graphnode] {$t_{2}$};
\node (s1) at (1*\xFactor,1*\yFactor) [graphnode] {$s_{1}$};
\node (t1) at (2*\xFactor,1*\yFactor) [graphnode] {$t_{1}$};

\draw(s2) to node [dnode,above] () {1}(v1);
\draw(v1) to node [dnode,above] () {1}(t2);
\draw(v1) to node [dnode,left] () {1}(s1);
\draw(s1) to node [dnode,above] () {1}(t1);
\end{tikzpicture}
\caption{Example where the weights of the agents in the order they are transporting a message is increasing.}\label{fig:example_increasing_weights}
\end{figure}
\end{example}

\subsection*{Upper bound on the benefit of collaboration}

\BoCUB*

\begin{proof}
We can assume without loss of generality that in the optimal schedule $\OPT$ every message $i$ is transported on a simple path from its starting point $s_i$ to its destination $t_i$. This can be easily achieved by letting agents drop the messages at intermediate vertices if they would otherwise transport it in a cycle. We define the directed multigraph $G_S=(V, E \dotcup \overline{E})$ as follows:
\begin{itemize}
\item $V$ is the set of vertices of the original graph $G$.
\item For every time in the optimal schedule that an agent traverses an edge $\{u,v\}$ from $u$ to $v$ while carrying a set of messages $M'$, we add the arc $e=(u,v)$ to $E$ and $\bar{e} =(v,u)$ to $\overline{E}$. We further label both edges with the set of messages $M'$ and write~$M_e = M_{\bar{e}} = M'$ to denote these labels. We call the edges in $E$ \emph{original} edges and the edges in $\overline{E}$ \emph{reverse} edges. 
\end{itemize}
We say that the tour of an agent $A$ \emph{satisfies the edge labels}, if every original edge $e \in E$ is traversed at most once by~$A$ and only while carrying the \emph{exact} set of messages $M_e$, and every reverse edge $\bar{e} \in \overline{E}$ is traversed by $A$ at most once and without carrying any message.

We will show that there exists a Eulerian tour satisfying the edge labels of every connected component of $G_S$. We then let the cheapest agent in each connected component follow the respective Eulerian tour. This agent traverses every edge exactly twice as often as the edge is traversed in the optimal schedule~$\OPT$ by all agents. As we choose the cheapest agent in each connected component, we obtain a schedule $S$ with $\cost(S)\leq 2 \cdot \cost(\OPT)$. 

By only considering a subset of the messages and a subschedule of~$\OPT$, we may from now on assume that~$G_S$ is strongly connected (by construction, every connected component of $G_S$ is strongly connected). Further, let $a_{\text{min}}$ be an agent with minimum cost among the agents that move in~$\OPT$, let $M(v)$ be the set of messages currently placed on vertex $v$, and let $M(a_{\text{min}})$ be the set of messages currently transported by agent~$a_{\text{min}}$. We first show that the procedure \textproc{computeTour} computes a closed tour for~$a_{\text{min}}$ that satisfies the edge labels, and afterwards we explain how we can iterate the procedure to obtain a Eulerian tour satisfying the edge labels.

\begin{algorithm}[t]
\caption{computeTour}
\begin{algorithmic}
\Function{computeTour}{}
	\State drop all messages 	
	\If {$\exists$ edge $e \in E$ incident to current vertex}
		\If{$M($current vertex$) \supseteq M_e$}
			\State pickup messages $M_e$, traverse $e$ and delete it
		\Else 
			\State let $j \in M_e \setminus M($current vertex$)$ 
		 	\State \textproc{fetchMessage}($j$, currentVertex)
		\EndIf
	\ElsIf{$\exists$ edge $\overline{e}\in \overline{E}$}
		\State traverse $\overline{e}$ and delete it 
	\EndIf	 
\EndFunction
\end{algorithmic}
\end{algorithm}

\begin{algorithm}[t]
\caption{fetchMessage}
\begin{algorithmic}
\Function{fetchMessage}{$i, v$}
	\State drop all messages 	
	\While{$i \notin M($current vertex$)$}
	  \If{there is a edge $\bar{e} \in \overline{E}$ with $i \in M_e$ leaving the current vertex }
		  \State traverse $\bar{e}$ and delete it 
		\Else
		  \State {\bf give up}
		\EndIf
	\EndWhile
	\While{$v \neq$ current vertex}
		\State let $e\in E$ be edge incident to current vertex with $i \in M_e$
		 \If {$M($current vertex$) \cup M(a_{\text{min}}) \supseteq M_e$}
		 	\State pickup messages $M_e$, drop all other messages, traverse $e$ and delete it
		 \Else
		 	\State let $j \in M_e \setminus (M($current vertex$) \cup M(a_{\text{min}}))$
		 	\State \textproc{fetchMessage}($j$, currentVertex)
		 \EndIf
	\EndWhile	 
\EndFunction
\end{algorithmic}
\end{algorithm}

\medskip
\underline{Claim 1:} If the agent $a_{\text{min}}$ starts in a vertex $v_0$ and follows the tour computed by \textproc{computeTour}, it satisfies the edge labels in every step and returns to its starting location.
\medskip

Both procedures \textproc{computeTour} and \textproc{fetchMessage} make sure that the agent traverses every edge $e \in E$ with label $M_e$ while carrying the exact set of messages $M_e$ and every edge $\bar{e} \in \overline{E}$ while carrying no messages. So the first part of the claim is clear. We only need to show that~$a_{\text{min}}$ cannot get stuck at some vertex $v^*$ before returning to $v_0$. As $G_S$ is Eulerian (ignoring the labels) and edges are deleted once they are traversed, this can only happen if some call \textproc{fetchMessage}$(i, v)$ gives up at vertex~$v^*$.
This means that the current vertex $v^*$ does not contain message $i$ and has no edge $\bar{e} \in \overline{E}$ with a label containing message~$i$. 

Note that when $a_{\text{min}}$ is currently proceeding according the call \textproc{fetchMessage}$(i, v)$, then it will be on a vertex of the path that message~$i$ takes from its start $s_{i}$ to its destination $t_{i}$ in the optimal schedule $\OPT$, and this path is simple by our initial assumption. 
Also note that, since edge labels are obeyed, message~$i$ only ever moves forward along its path in~$\OPT$.
This means that if~$a_{\text{min}}$ is stuck at vertex~$v^*$, there must initially have been an edge $\bar{e}=(v^*,w) \in \overline{E}$ incident to~$v^*$ with $i \in M_e$ that was taken by the agent earlier and then deleted. The agent traverses edges in $\overline{E}$ only in the procedure \textproc{fetchMessage}. So there must have been a call \textproc{fetchMessage}$(j_1, v_1)$ before, where the agent traversed the edge $\bar{e}$. This call cannot have been completed as otherwise the original edge $e=(w,v^*) \in E$ corresponding to $\bar{e}$ would have been used by~$a_{\text{min}}$ and the message $i$ would have already reached $v^*$, since~$i\in M_e$. This contradicts that the message path of $i$ is simple and the agent is currently proceeding according the call \textproc{fetchMessage}$(i, v)$ at vertex $v^*$.

As the call \textproc{fetchMessage}$(j_1, v_1)$ is not complete, there must be a vertex $v_2$ and a message $j_2$ missing at this vertex to further carry $j_1$ on its paths to the destination, and a call \textproc{fetchMessage}$(j_2, v_2)$ which is also
 incomplete. By iterating this argument, we obtain that the current stack of functions is \textproc{fetchMessage}$(j_s, v_s), \ldots,$ \textproc{fetchMessage}$(j_1, v_1)$ for some $s \in \mathbb{N}$, where $j_s=i$ and $v_s=v$. In the optimal schedule $\OPT$ the message $j_2$ needs to be transported to $v_2$ before $j_1$ can be further transported from
  $v_2$ together with $j_2$. Similarly, message $j_r$ needs to be transported to $v_r$ before message $j_{r-1}$ can be transported further together with message $j_r$ from $v_r$ for $r=2,\ldots, s$. Moreover, on the edge $e=(w,v^*)$ the messages $j_1$ and
   $j_s=i$ need to be transported together and in particular, message $j_1$ needs to be 
   transported together with $j_s$ before $j_s$ can be transported further. But in the optimal schedule the messages must be transported in a certain sequence and it cannot be that message $i$ needs to be transported to $v$ before messages $j_1$ is transported to message $v_1$ and vice versa. Thus. \textproc{computeTour} must terminate with $a_{\text{min}}$ returning to the starting location~$v_0$.

\bigskip
\underline{Claim 2:} After completing a tour given by \textproc{computeTour} the following holds: Every message $i$ has either been transported to its destination or it is on a vertex $v_i$ such that there is a path from $v_i$ to $t_i$ with edges in $E$ containing~$i$ in their labels, and a path in the reverse direction with edges in $\overline{E}$ containing~$i$ in their labels.
\bigskip

Every edge $e \in E$ with label $M_e$ is only traversed if the agent~$a_{\text{min}}$ carries the set of messages $M_e$. Thus at any time there is a path from the current location~$v_i$ of message~$i$ to its destination~$t_i$ with edges containing~$i$ in the label. This shows the first part of the claim.

Observe that a completed call \textproc{fetchMessage}$(i,v)$ yields a closed walk, as the agent starts and ends in $v$. Moreover, it first traverses exactly all edges in $\overline{E}$ on the path from~$v$ to the current position~$v_i$ of message~$i$ and then all edges in~$E$ on the path from~$v_i$ to~$v$. Inductively, this also holds for all levels of recursive calls of \textproc{fetchMessage}. Hence, for every original edge $e \in E$ also the corresponding reverse edge $\bar{e} \in \overline{E}$ is traversed in a call of \textproc{fetchMessage}.

This fact also implies that for any edge~$E \dotcup \overline{E}$ traversed in the procedure \textproc{computeTour} (and not in \textproc{fetchMessage}), the corresponding original/reverse edge cannot be traversed in a call of \textproc{fetchMessage}. Inductively, we can therefore argue that if~$e_1,e_2,\ldots, e_s  \in E \dotcup \overline{E}$ are the edges traversed in the procedure \textproc{computeTour} in this order, such that the corresponding original/reverse edges were not traversed, then~$e_s,\ldots, e_1$ is a path from the current location of~$a_{\text{min}}$ to its starting vertex. This shows that at termination for every original edge $e \in E$ also the corresponding reverse edge $\bar{e} \in \overline{E}$ is traversed.

\medskip
\underline{Claim 3:} Combining the tours returned by multiple calls of \textproc{computeTour} yields a Eulerian tour that satisfies the edge labels in every step.
\medskip

Assume that the tour~$T$ resulting from a call of \textproc{computeTour} does not traverse all edges of~$G_S$. Let~$v_0$ be the starting vertex of the tour, $v$ be the last vertex on the tour that is is incident to an edge which is not visited, and~$v_i$ be the position of message~$i$ after the tour~$T$ is finished. Further, let~$G'_S$ be the graph~$G_S$ after the call of \textproc{computeTour}, i.e., without all edges in~$T$ and message~$i$ at position~$v_i$ instead of~$s_i$. We want to show that we can run \textproc{computeTour} on~$G'_S$ with~$a_{\text{min}}$ starting at~$v$ and then add the resulting tour~$T'$ to~$T$ as follows: First~$a_{\text{min}}$ follows~$T$ until the last time it visits~$v$, then it follows~$T'$, and finally the remaining part of~$T$.

The graph~$G_S'$ is a feasible input to \textproc{computeTour} by Claim~2. It corresponds to the instance of \ourProblem, where all message transported in the schedule $\OPT$, which are done in the tour~$T$ by~$a_{\text{min}}$, are completed, and the agent positions are adapted accordingly. By Claim~1, \textproc{computeTour} will produce a tour that satisfies the edge labels. The only problem that can occur when combining the tours~$T$ and~$T'$ therefore is that
during following the tour~$T'$, \textproc{fetchMessage}$(i,v)$ is called, but some message~$i$ is not yet transported to~$v_i$ because the tour~$T$ has not been completed. This means that vertex~$v_i$ is visited after the last time~$v$ is visited by the tour~$T$. By the choice of~$v$, all edges incident to~$v_i$ must be visited by the tour $T$, in particular, we must have that~$v_i=t_i$ and message~$i$ is delivered to its destination by the tour $T$. But then~$G'_S$ does not contain any edge with label~$i$ by Claim~2.

By iterative applying the above argument, we obtain a Eulerian tour that satisfies the edge labels in every step.
\end{proof}

\subsection*{Upper bound on the benefit of collaboration for a single message}

\BoConemessageUB*

\begin{proof}
	By using Dijkstra's algorithm, we can determine the agent that can transport the message from $s$ to $t$ with lowest cost. We need to show that this is at most $1 / \ln(2)$ the cost of an optimum using all agents.	

	Fix an optimum solution and let the agents $a_1,a_2,\dots,a_r$ be labeled in the order in which they transport the message in this optimum solution (ignoring unused agents). 
	We can assume by Lemma~\ref{lemma:non-increasing} that $w_1 >w_2 > \ldots >w_r$. 
	By scaling, we can further assume without loss of generality that $w_r=1$ and that the total distance traveled by the message is 1.
	Now, for each point $x \in [0,1]$ along the message path there is an agent $a_j$ with cost $w_j$ carrying the message at this point in the optimum schedule and we can define a function $f$ with $f(x)=w_j$. 
	The function $f$ is a step function that is monotonically decreasing by Lemma~\ref{lemma:non-increasing} with $f(0)=w_1$ and $f(1)=w_r=1$.  
	We now choose the largest $b \in [0,1]$ such that $f(x) \geq \tfrac{b}{x+1}$, see Figure~\ref{fig:benefit-of-collaboration}.

	\begin{figure}[t!]
		\vspace{-4ex}
		\centering
		\usetikzlibrary{arrows,intersections}
		\tikzstyle{gLine}=[thick]
		\tikzstyle{kreis}=[circle,inner sep=0pt, minimum size=3pt,draw,thick,color=black,fill=white]
		\tikzstyle{niceArrow}=[->,>=stealth',
		    dot/.style = {draw,
		      fill = white,
		      circle,
		      inner sep = 0pt,
		      minimum size = 4pt
		    }]
		\begin{tikzpicture}[scale=0.75]
			\coordinate (O) at (0,0);
			\draw[gLine,niceArrow] (-0.3,0) -- (8,0) coordinate[label = {below:$x$}] (xmax);
			\draw[gLine,niceArrow] (0,-0.3) -- (0,3) coordinate[label = {}] (ymax);
			\draw[gLine] (0,2.5)--(1.5,2.5);
			\node (k1) at (1.5,2.5) [kreis]{};
				\draw[gLine] (1.5,2.1)--(2.5,2.1);
			\node (k2) at (2.5,2.1) [kreis]{};
				\draw[gLine] (2.5,1.65)--(3,1.65);
			\node (k3) at (3,1.65) [kreis]{};  
				\draw[gLine] (3,1.4)--(5.5,1.4);
			\node (k3) at (5.5,1.4) [kreis]{};  
				\draw[gLine] (5.5,1.1)--(7,1.1);
		
			\draw[gLine]  (7,0.1)--(7,-0.1);
			\node (l1) at (7,-0.4) [scale=1]{$1$};
		
			\draw[dashed] (3,1.4)--(3,-0.1);
			\draw[gLine]  (3,0.1)--(3,-0.1);
			\node (l2) at (3,-0.4) [scale=1]{$x^*$};
		
			\draw[dashed] (3,1.4)--(0,1.4);
			\draw[gLine]  (-0.1,1.4)--(0.1,1.4);
			\node (l3) at (-0.5,1.4) [scale=1]{$w_{j^*}$};
		
			\node (l3) at (5,0.7) [color=blue, scale=1]{$\frac{b}{x+1}$};
			\node (l3) at (3,2.5) [, scale=1]{$f(x)$};
			\node (i4) at (3,1.4) [circle,inner sep=0pt, minimum size=4pt, fill,color=blue]{};
		
			\draw[thick,scale=4,domain=-0:2,smooth,variable=\x,blue] plot ({\x},{0.61/(\x+1)});
		
			\path[name path=x] (0.3,0.5) -- (6.7,4.7);
		
			\path[name path=y] plot[smooth] coordinates {(-0.3,2) (2,1.5) (4,2.8) (6,5)};
		\end{tikzpicture}
		\caption{Choosing the largest $b$ such that $\tfrac{b}{x+1}$ is a lower bound on the step-function $f$ representing the weight of the agent currently transporting the message.}
		\label{fig:benefit-of-collaboration}
	\end{figure} 

 	Note that $b \geq 1$ as $f(x) \geq 1 \geq \frac{b}{x+1}$ for $b=1$ and all $x \in [0,1]$. 
 	Further, let $g_{j}$ be the distance traveled by agent $a_j$ without the message and $g:=\sum_{j=1}^r g_j w_i$ the total cost for the distances traveled by all agents without the message. 
 	We obtain the following lower bound for an optimum solution
 	\begin{linenomath}
 	\[	\cost(\OPT) = \int_0^1 f(x)\,\mathrm{d}x + g \geq \int_0^1 \frac{b}{x+1}\,\mathrm{d}x +g = b \ln(2) + g.	\]
 	\end{linenomath}
 	By the choice of $b$, the functions $f(x)$ and $\frac{b}{x+1}$ coincide in at least one point in the interval $[0,1]$.
	Let this point be $x^*$ and $a_{j^*}$ be the agent carrying the message at this point. This means that $f(x^*)= \frac{b}{x^*+1} = w_{j^*}$. 
	We will show that it costs at most $(1/ \ln(2))\cost(\OPT)$ for agent $a_{j^*}$ to transport the message alone from $s$ to $t$. 
	The cost for agent $a_{j^*}$ to reach $s$ is bounded by $g_{j^*} w_{j^*} + x^* \cdot w_{j^*}$ and the cost for transporting the message from $s$ to $t$ is bounded by $w_{j^*}$. 
	Thus, the cost of the algorithm using only one agent can be bounded by
	\begin{linenomath}
	\[	\cost(\textsc{ALG})\leq g_{j^*} w_{j^*} + x^* \cdot w_{j^*} + w_{j^*} = g_{j^*}w_{j^*} + (x^* +1) \cdot \frac{b}{x^*+1} = b + g_{j^*} w_{j^*}.	\]
	\end{linenomath}
	By using that $g_{j^*} w_{j^*} \leq g$, we finally obtain 
	$
	\tfrac{\cost(\ALG)}{\cost(\OPT)}\leq \tfrac{b + g_{j^*} w_{j^*}}{b \ln(2) +g} \leq
	\tfrac{b}{ b \ln(2)}  = \tfrac{1}{\ln(2)}.
	$
\end{proof}

\subsection*{Upper bound on the benefit of collaboration without intermediate dropoffs}

\nointermediate*

\begin{proof}
	For $\underline{\kappa = \infty}$ this is a corollary of Theorem~\ref{theo-ub-boc}: an agent $a$ with infinite capacity can just as well keep a message $m_i$ once it was picked up, 
	i.e. we can simply remove all the actions for this message between the first pick-up $(a,s_i,m_i,+)$ and the last drop-off $(a,t_i,m_i,-)$.

	For $\underline{\kappa = 1}$ we need a different analysis. Given an optimum schedule $\OPT$, we look at how the trajectories of the messages are connected by the agents.
	More precisely, we construct an auxiliary multigraph $G' = (V', E')$ as follows:
	The vertex set $V'$ consists of all messages. Then, for each agent $a$ we look at its subsequence of the optimum schedule, $\OPT|_{a}$. 
	Since $a$ has capacity 1, its subsequence consists of alternating drop-offs $(a,p,m_i,-)$ and pick-ups $(a,q,m_j,+)$.
	For each drop-off followed by a pick-up action we add an edge $(m_i,m_j)$ of length $d_G(p,q)$ to $E'$ (Figure~\ref{fig:boc-nointermediate}).
	These edges correspond to the portions of the optimum schedule where the agent travels without carrying a message.
	\begin{figure}[ht!]
		\centering
		\includegraphics[width=\linewidth]{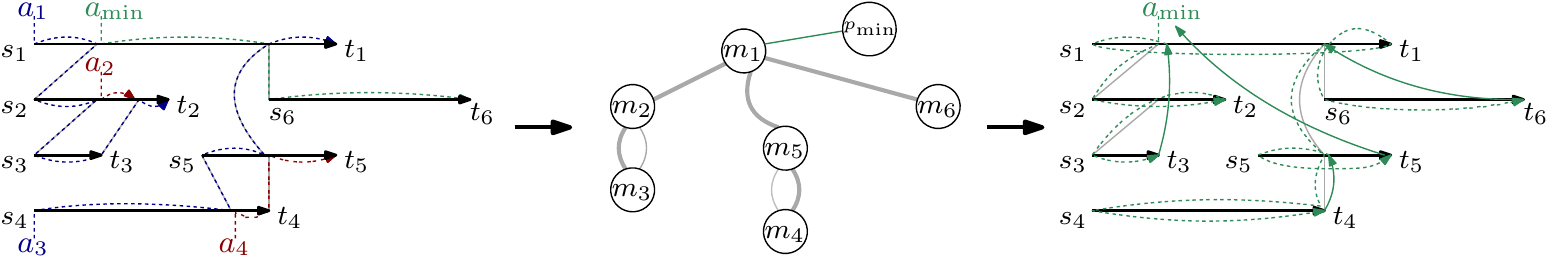}
		\caption{(left) An optimal schedule, (center) the auxiliary graph, (right) the 2-approximation.}
		\label{fig:boc-nointermediate}
	\end{figure}
	We assume without loss of generality that $G'$ is connected and denote by $a_{\min}$ the agent of minimum weight involved in $\OPT$ 
	(otherwise we can look at each connected component and its agent of minimum weight separately). 
	Assume that the first action of~$a_{\min}$ in~$\OPT$ is to move from its starting position~$p_{\min}$ to a node~$p$ where it picks up message~$m_1$ 
	(note that $p$ can potentially lie anywhere on the trajectory between $s_1$ and $t_1$). 
	We model this by adding a node~$p_{\min}$ to~$V'$ and connecting it to~$m_1$ by an edge of corresponding length~$d_G(p_{\min}, p)$.
	Now we can take a minimum spanning tree of $G'$ and remove all redundant edges. 
	
	Note that the total length of the minimum spanning tree is a lower bound on the sum of the distances traveled by all the agents in $\OPT$ \emph{without} carrying a message.
	Thus any schedule~$S$ of~$a_{\min}$ which move exactly twice along the trajectory of each message \emph{and} twice along the path corresponding to each edge of the minimum spanning tree of $G'$ has a cost of at most $\cost(S) \leq 2\cdot \cost(\OPT)$. The following tour satisfies this property and delivers each message to its destination immediately after it is picked up:

	We first let~$a_{\min}$ walk to~$p$, from where~$a_{\min}$ proceeds towards~$s_1$. 
	When~$a_{\min}$ reaches~$s_1$, it picks up~$m_1$ and delivers it to~$t_1$ along its trajectory in~\OPT.
	Once $a_{\min}$ reaches $t_1$, we let it return to~$p$ and from there back to its original position~$p_{\min}$.
	If, however, along its way from $p$ to $s_1$ or from $t_1$ to $p$ the agent visits an endpoint of a path corresponding to an edge of the minimum spanning tree,
	we first let $a_{\min}$ serve the adjacent subtree recursively, see Figure~\ref{fig:boc-nointermediate}. 
	It is easy to see that in the resulting schedule~$S$ every message is directly transported from its source to its destination, and that the capacity is respected  at all times. 
\end{proof}
\section{Appendix: Planning}
\label{app:planning}

\singleagenthardness*

\begin{proof}
We proceed by a reduction from Hamiltonian cycles on a grid graph, a problem shown to be \NP-hard by Itai et al.~\cite{NPonGrid}.
A similar reduction was used for a sorting problem by Graf~\cite{ESA15}.
Given an unweighted grid graph $H = (V_H, E_H)$ with $V_H = \left\{ v_1, v_2, \ldots, v_n \right\}$. We add to every vertex $v$ a new vertex $v'$ with an edge $e= (v,v')$ and a
message with start $v$ and target $v'$. Denote the new graph with $G = (V, E)$, where $V = V_H \cup V_H'$, $V_H' = \{v' \mid v \in V_H\}$ and $E =
E_H \cup E_H'$, where $E_H' = \left\{ (v,v')\ | \ v \in V_H \right\}$. 

\begin{figure}[b!]
	\centering
	\includegraphics[width=\linewidth]{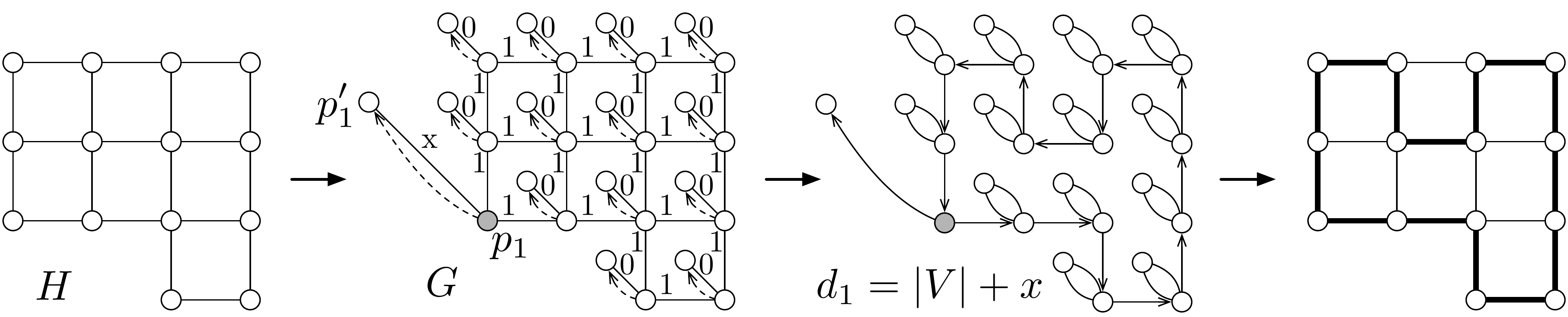}
	\caption{Finding a Hamiltonian cycle via \ourProblem with a single agent.}
	\label{fig:hamilton-cycle-reduction2}
\end{figure}

We build an instance of \ourProblem by taking $G$ and placing a single agent on an arbitrary vertex $p_1 \in V_H$. We let $w_1 = 1$ and set unit edge length $l_e = 1$ for all edges in $E_H$ and edge length $l_e=0$ for all edges in $E_{H'}$ except $(p_1,p_1')$, see Figure~\ref{fig:hamilton-cycle-reduction2}. The edge $(p_1,p_1')$ gets length $x = l_{(p_1,p_1')}=|V|$ instead.

Now let the $d_1$ be the length of the shortest path in $G$ starting in $p_1$ on
which the agent can deliver all messages.
We now argue that there is a Hamiltonian cycle in $H$ if and only if $d_1 = 2|V|$.
To see that $2|V|$ is a lower bound for $d_1$, let us distinguish whether the agent ends at $p_1'$ or not.
If we end at $p_1'$, we have to reach every $v \in V_H\setminus \{p_1\}$ at least once and also go back to $p_1$ before using $(p_1,p_1')$. This sums up to at least $(|V|-1)+1+|V|= 2|V|$. If we end somewhere else, we have to use $(p_1,p_1')$ twice, hence $d_1 \geq 2|V|$. So we get a schedule of cost $2|V|$ if and only if we reach every vertex exactly once and end at $p_1'$. When removing all the $E_H'$-steps, such a schedule directly corresponds to a Hamiltonian cycle as illustrated in Figure~\ref{fig:hamilton-cycle-reduction2}.
\end{proof}

\approximationhardness*

\begin{proof}
	We build on top of a result by Karpinski, Lampis and Schmied~\cite[Theorem 4]{karpinski2015new} which shows that the symmetric, metric traveling salesperson problem is hard to approximate with ratio better than $\frac{123}{122}$. For a reduction, we take any metric, undirected graph $H$, duplicate the vertices and put a zero-length edge and a single message between each of them, just like in Theorem~\ref{thm:singleagenthardness} / Figure~\ref{fig:hamilton-cycle-reduction2}.
	To find a suitable length $x$ for the extra edge $(p_1,p_1')$ at the (arbitrary) starting vertex we have to consider the following: In a traveling salesman tour in $H$, we want the agent to come back to $p_1$ in the end.
	Hence in \ourProblem on $G$ we want the agent to end at $p_1'$. We achieve this by setting $x$ large enough to avoid traveling $(p_1,p_1')$ twice.
	Let $M$ be the cost of a minimum spanning tree in $H$.
	Clearly, both the optimum traveling salesman path and the optimum traveling salesman tour have cost at least $M$ but also at most $2M$.
	Hence, setting the length $x$ of the extra edge $(p_1,p_1')$ to $2M$ ensures that any schedule for \ourProblem on $G$ which doesn't end in $p_1'$ (and thus uses $(p_1,p_1')$ twice) has cost at least $2\cdot 2M+M=5M$, 
	while an optimum schedule has delivery cost at most $2M+2M=4M$. 
	It remains to look at schedules which end in $p_1'$: 
	As the extra edge contributes at most two thirds of the cost of the final schedule, at least one third of the approximation gap is conserved, giving $1 + \frac{1}{3\cdot122} = \frac{367}{366}$.
\end{proof}

\restrictedplanning*

\begin{proof}(\mathversion{bold}$\kappa=\infty$\mathversion{normal}).
	By the given restriction, separate planning of each $S^R|_{a_j}$ independently maintains feasibility of $S^R$. We denote by $m_{j1}, m_{j2}, \ldots, m_{jx}$ the messages appearing in $S^R_{a_j}$. 
	We define a complete undirected auxiliary graph $G'=(V',E')$ on the node set 
	$V' = \left\{ p_j \right\} \cup \left\{ s_{j1}, s_{j2}, \ldots, s_{jx} \right\} \cup \left\{ t_{j1}, \ldots, t_{jx} \right\}$ with edges $(u,v)$ having weight $d_G(u,v)$.

	For $\underline{\kappa=\infty}$, the schedule $\OPT(S^R)|_{a_j}$ corresponds to a Hamiltonian path $H$ in $G'$ of minimum length, starting in $p_j$, 
	subject to the condition that for each message $m_{ji}$ its source $s_{ji}$ is visited before its destination $t_{ji}$.
	We approximate $\OPT(S^R)|_{a_j}$ with a schedule $\ALG(S^R)|_{a_j}$ that first collects all messages $m_{j1}, \ldots, m_{jx}$ before delivering all of them. 
	We start by computing a minimum spanning tree $MST'$ of the subgraph of $G'$ consisting of all nodes $\left\{ p_j \right\} \cup \left\{ s_{j1}, \ldots, s_{jx} \right\}$.
	A DFS-traversal on $MST'$ starting from $p_j$ collects all messages, returns back to $p_j$ and has cost $2\cdot \sum_{e \in E(MST')} l_e \leq 2\cdot \cost(\OPT(S^R)|_{a_j})$. 
	Next we consider the subgraph of $G'$ consisting of all nodes $\left\{ p_j \right\} \cup \left\{ t_{j1}, \ldots, t_{jx} \right\}$.
	Using Christofides' heuristic for the metric TSP path version with fixed starting point $p_j$ and arbitrary endpoint, due to Hoogeveen~\cite{Christofides76, Hoogeveen91}, 
	we can deliver all messages by additionally paying at most $\tfrac{3}{2}\cdot \cost(\OPT(S^R)|_{a_j})$. In total, this gives a $3.5$-approximation.
\end{proof}

\section{Appendix: Coordination}
\label{app:coordination}

\begin{restatable}[Energy cost of a 3SAT-solution]{lem}{solutioncostlemma}
	Given a satisfiable assignment (a solution) for the variables of a 3CNF $F$ there is a feasible schedule $\SAT$ of the agents in $G(F)$ which
	has a total (energy) cost of $\cost(\SAT) = 4xy + 2y + x(y^2+y+1)\varepsilon$.
	\label{lem:optimum-cost}
\end{restatable}
%

\begin{proof} We are given $x$ variables, $y$ clauses and a satisfiable assignment of the variables. We construct a schedule from
	the assignment as follows: Assume that variable $v$ is set to true and consider the corresponding variable box. 
	\begin{itemize}
		\item	The variable agent (which has weight 1) delivers all messages on the full $\mathit{true}$-path (which has length
			$2y+\varepsilon$).  The energy needed to do so is $2y+\varepsilon$.
		\item	Each literal agent placed on a node $v_{\mathit{false},i}$ transports two messages: the message with source
			$v_{\mathit{false},i}$ and the message with source $v_{\mathit{false},i+1}$. Summing over all messages on the
			$\mathit{false}$-path we need an energy of $2\cdot \left( (1+y\varepsilon) +\ldots + (1+ \varepsilon) \right)$.
	\end{itemize}
	Hence for the messages in each of the $x$ variable boxes we have an energy consumption of
	$2y+\varepsilon+2y+2\sum_{i=1}^y{i\varepsilon} = 4y+(y^2+y+1)\varepsilon$. Furthermore, since we start from a satisfiable assignment, the
	source of each clause message is connected to at least one not yet used agent of weight $1+i\varepsilon$. Such an agent is adjacent to the
	source of that clause message only and can move to the source (distance $\tfrac{1-i\varepsilon}{1+i\varepsilon}$), pick it up and deliver it
	(distance $1$), hence for each clause we get an energy consumption of exactly $(1+i\varepsilon)( \tfrac{1-i\varepsilon}{1+i\varepsilon} +1 ) =
	2$.
	Summed over all variables and all clauses we get a total energy consumption of $x \cdot \left( 4y+(y^2+y+1)\varepsilon \right) + 2y = 4xy + 2y
	+ x(y^2+y+1)\varepsilon$. 
\end{proof}

\subparagraph{Fixed sequence (schedule without agent assignment).} We now fix a sequence $S^-$ that describes the schedule constructed in Lemma~\ref{lem:optimum-cost} but which does not allow us to infer a satisfiable assignment. 
In our sequence $S^-$ every pick-up action of one of the $4xy+y$ many messages $m_i$ at its location $s_i$ is immediately followed by its drop-off at its destination $t_i$. 
Hence we are restricted to a schedule
$(\text{\textunderscore}, s_{\pi(1)}, m_{\pi(1)}, +), (\text{\textunderscore}, t_{\pi(1)}, m_{\pi(1)}, -), (\text{\textunderscore}, s_{\pi(2)}, m_{\pi(2)}, +), \ldots,$ $ (\text{\textunderscore},t_{\pi(4xy+y)}, m_{\pi(4xy+y)}, -)$,
where $\pi$ can be any permutation on $1,\ldots, 4xy+y$ that satisfies the following property:
If two messages $m_i,m_j$ lie on the same $\mathit{true}$- or $\mathit{false}$-path originating at a variable node $v$, then $d_G(F)(v, s_i) < d_G(F)(v, s_j) \Rightarrow \pi^{-1}(i) < \pi^{-1}(j)$
(meaning if $m_i$ lies to the left of $m_j$, it should precede $m_j$ in the schedule).

\subparagraph{Energy Consumption of Optimal Schedules.} In the next three Lemmata~\ref{lem:variablebox-independence},~\ref{lem:agent-movement}
and~\ref{lem:optimum-schedule-cost} we show that (i) the total energy consumption of any optimum schedule \OPT\ is $\cost(\OPT) \geq \cost(\SAT)$,
and (ii) in every optimum schedule with $\cost(\OPT) = \cost(\SAT)$, each variable agent delivers exactly all messages on either the $\mathit{true}$- or the
$\mathit{false}$-path. 
We remark that (i) is true independent of whether $\OPT$ adheres tho the fixed schedule without assignments $S^-$ or not and holds for any capacity $\kappa$. 
In the case of (ii), $\OPT$ has exactly the properties as described in the proof of Lemma~\ref{lem:optimum-cost}.
Since thus for each agent $a_j$ the subsequence $\OPT|_{a_j}$ is uniquely determined, and since each message is transported by a single agent (without intermediate drop-offs),
these subsequences can be merged to match the prescribed fixed order of the actions $S^-$.

\begin{figure}[t!]
	\centering
	\includegraphics[width=\linewidth]{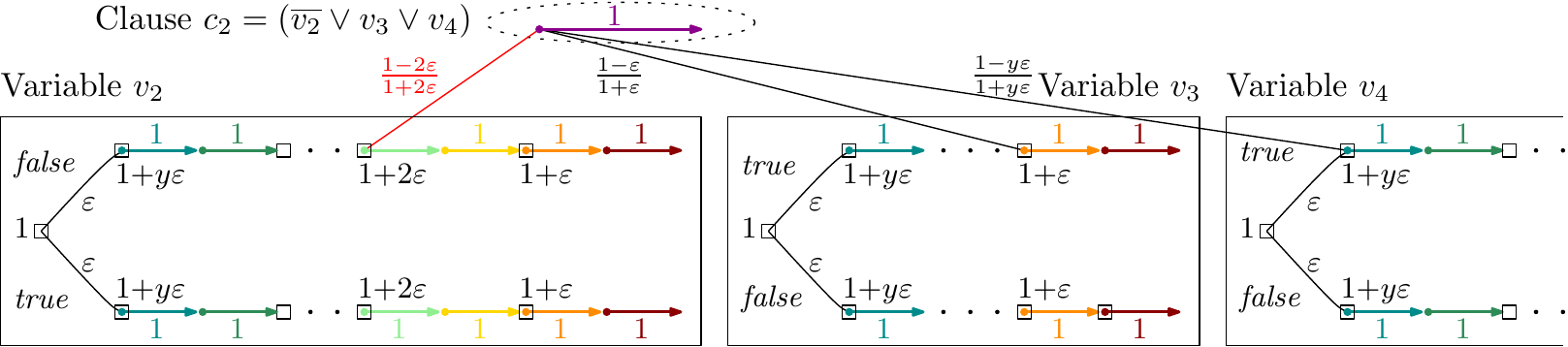}
	\caption{Agent positions ($\square$) and weights (in black); Messages ($\rightarrow$) and edge lengths (in color).}
	\label{fig:weight-hardness2}
\end{figure}

\subparagraph{Lower Bound.} To this end we first give an upper bound $\mathit{UB}$ on $\cost(\SAT)$ and a lower bound $\mathit{LB}$ on the total energy
consumption of any feasible schedule $S$. First, since $\varepsilon = (8xy)^{-2}$, we define $\mathit{UB} := 4xy + 2y + 0.25 > 4xy + 2y +
x(y^2+y+1)\varepsilon = \cost(\SAT)$. For the lower bound, note that every agent has weight at least 1. We double count the distance traveled by the
agents via the distance covered by the messages, hence we have to be careful to take into account that an agent might carry two or more messages at
the same time (we do not want to count the energy used during that time twice or more). Hence for each of the $4xy + y$ messages we count the last
edge over which it is transported \emph{towards} its target. All message targets (and thus the distance traveled towards them) are disjoint and have
adjacent edges all of length at least $\tfrac{1-y\varepsilon}{1+y\varepsilon} > 1-2y\varepsilon > 1-\tfrac{1}{32xy}$. Additionally, before an agent
can deliver a \emph{clause message}, it needs first to travel towards its source. Such edge crossings are not counted yet, hence we can add an
additional distance of $\tfrac{1-y\varepsilon}{1+y\varepsilon}$ for each clause. Overall we get a lower bound on the total energy consumption of
$\mathit{LB} := 4xy + 2y - 0.25 < (4xy + y + y)\cdot ( 1 -\tfrac{1}{32xy} )$.

\begin{restatable}[Independence of variable boxes]{lem}{boxindependencelemma}
	Any optimum schedule \OPT, in which an agent placed in a variable box $u$ moves to a clause node and back into a variable box of some 
	(not necessarily different) variable $v$, has $\cost(\OPT) > \cost(\SAT)$.
	\label{lem:variablebox-independence}
\end{restatable}
\begin{proof} 
	Assume for the sake of contradiction that an agent leaves its variable box, walks to a clause node and later on back to a variable box. If
	this agent delivers the clause message, it has to return and thus walks the corresponding clause message's distance of $1$ twice. If it doesn't,
	then we haven't included the travel \emph{towards} the clause node yet. In both cases, we add at least another $1 -\tfrac{1}{32xy} > 0.5$ to
	the given lower bound $\mathit{LB}$, yielding $\cost(\OPT) > \mathit{LB} + 0.5 = \mathit{UB} > \cost(\SAT)$.
\end{proof}

From now on we can restrict ourselves to look at feasible schedules where each each agent either stays inside its variable box, or walks to deliver a
clause message and stays at the target of that clause message. Next we show that we can also assume that agents walk only from left to right inside a
$\mathit{true}$- or $\mathit{false}$-path:

\begin{restatable}[Agents move from left to right]{lem}{lefttorightlemma}
	Any optimum schedule \OPT, in which there is an agent $a$ that moves at some point in the schedule from right to left along a $\mathit{true}$- or
	$\mathit{false}$-path, has $\cost(\OPT) > \cost(\SAT)$.
	\label{lem:agent-movement}
\end{restatable}
\begin{proof} 
	By Lemma~\ref{lem:variablebox-independence} we can restrict ourselves to schedules \OPT\ where each message $i$ has to be transported over the
	edge connecting $s_i$ and $t_i$. Without loss of generality we assume that no message $i$ ever leaves the interval $\left[ s_i, t_i \right]$
	and that it is transported from $s_i$ to $t_i$ monotonically from left to right (otherwise we could adjust the schedule accordingly by keeping
	the trajectories of the agents but adapting pick-up and drop-off locations).

	First assume for the sake of contradiction that in \OPT\ there is an agent $a$ whose trajectory contains moves from right to left of total
	length at least $0.5$. The energy needed to do this is at least $0.5$, and these moves are not yet included in the lower bound $\mathit{LB}$. As
	before, by adding $0.5$, we get $\cost(\OPT) > \mathit{LB} + 0.5 = \mathit{UB} > \cost(\SAT)$. Hence in the following we assume that each agent moves strictly
	less than $0.5$ from right to left.

	We are going to show that any such schedule \OPT\ can be transformed into a schedule of smaller cost, contradicting the optimality of \OPT. 
	Consider the longest consecutive right-to-left move of agent $a$:
	Since $a$ moves by less than $0.5$ to the left, it must come from a handover $h_i$ point lying inside an edge $(s_i, t_i)$ or go to a handover
	point $h_j$ lying inside an edge $(s_j, t_j)$.

	In the first case, $h_i$ is closer to $s_i$ than to $t_i$ and the previous action of $a$ in the schedule must have been $(a, h_i, i)$. The
	agent $b$ picking up message $i$ at $h_i$ must have its starting position on or to the left of $s_i$. Hence we could replace $a$'s pick-up and
	drop-off of message $i$ with pick-up and drop-off by agent $b$, thus strictly decreasing the total distance $d_a$ traveled by $a$,
	contradicting the optimality of $OPT$.
	In the second case, $h_j$ is closer to $t_j$ than to $s_j$ and agent $a$ moves to the right after picking up message $j$. Let $b$ denote the
	agent that dropped off $j$ at $h_j$. By the previous remarks we know that $b$ will not move to the left in its next action (if any),
	$b$ can't reach $s_j$ and no other message is inside $\left[ s_i, t_i \right]$.
	Furthermore by the weights given in our hardness reduction we know $w_b < 2w_a$. Therefore we can move $h_j$ by a small $\delta > 0$ to the right,
	thus strictly decreasing $w_a \cdot d_a + w_b \cdot d_b$, contradicting the optimality of $\OPT$.
\end{proof}

From now on we assume that agents do not move from right to left at all. Now we are ready to prove the key relation between optimum schedules and
\SAT\ schedules:
\begin{lemma}[Energy cost of an optimum schedule]
	Any optimum schedule \OPT\ either has total energy consumption $\cost(\OPT) > \cost(\SAT)$ or $\cost(\OPT) = \cost(\SAT)$. In the latter case, each
	variable agent either delivers exactly all messages on its $\mathit{true}$-path or exactly all messages on its $\mathit{false}$-path. Furthermore the literal
	agents on the respective other path deliver exactly two literal messages each. Finally, clause messages are only delivered by freed up literal
	agents on the paths chosen by the variable agents.
	\label{lem:optimum-schedule-cost}
\end{lemma}

\begin{proof} 
	By Lemmata~\ref{lem:variablebox-independence} and~\ref{lem:agent-movement} we may assume that no agent travels into another variable box and
	that agents only move from left to right on any $\mathit{true}$- or $\mathit{false}$-path. Furthermore we have seen in the proof of Lemma~\ref{lem:agent-movement}
	that this implies that every literal message is carried from its source to its target by a single agent in a continuous left-to-right motion.
	We now show that if there was a variable agent $a$ which does not deliver either all messages on its adjacent $\mathit{true}$-path or its adjacent $\mathit{false}$-path, 
	then we would get a contradiction to optimality:
	
	Assume first for the sake of contradiction that $a$ stays on its starting location. Then we can move $a$ to the first internal node of its
	$\mathit{true}$-path, which contributes an additional $\varepsilon$-distance to the total energy consumption. Let $b$ be the agent carrying the first
	literal message; we know that $b$ must have weight $1+y\varepsilon$. We replace $b$ in the schedule by $a$, saving at least $y\varepsilon$
	energy already on the first literal message, contradicting the optimality of the schedule.
	Now assume that $a$ either deviated at some internal node $v_{\mathit{true},2(y-i)}$ of the path it entered (to deliver a clause message) or that it 
	stopped at such an internal node ($v_{\mathit{true},2(y-i)}$ or $v_{\mathit{true},2(y-i)+1}$). 
	Let $b$ denote the agent carrying the message which was placed on the specified node. The edge to the adjacent clause (if any) has length
	$\tfrac{1-i\varepsilon}{1+i\varepsilon}$ and $b$ has weight $1+j\varepsilon$, with $j \geq i$. Now we can switch $a$ and $b$ in the remainder
	of the schedule. The potential increase of energy cost (on the clause message delivery) amounts to at most $((1+j\varepsilon)-1)\cdot
	(\tfrac{1-i\varepsilon}{1+i\varepsilon}+1) < 2j\varepsilon$, while the gained energy on the next two literal messages is at least
	$j\varepsilon \cdot 2$, again contradicting the optimality of the schedule.

	Hence each variable agent delivers either all messages on its $\mathit{true}$-path or all messages on its $\mathit{false}$-path. This allows us to give a new
	lower bound $\mathit{LB}'$ on the energy consumption $\cost(\OPT)$: Each variable agent contributes an energy consumption of $2y+\varepsilon$ to the total.
	Delivery of each message on the respective other path needs an agent with starting location coinciding with or to the left of the message
	source, yielding an energy contribution of at least $2\cdot ((1+y\varepsilon)+ \ldots + (1+\varepsilon)) = 2y + y(y+1)\varepsilon$, with
	equality if and only if each literal agent placed on $v_{\mathit{true},i}$ delivers the message placed on 
	$v_{\mathit{true},i}$ \emph{and} the consecutive message with source on $v_{\mathit{true},i+1}$.
	Finally the source of each clause message is reached by an agent of weight $1+j\varepsilon$ over an edge of length
	$\tfrac{1-i\varepsilon}{1+i\varepsilon}$, $j \geq i$, hence the delivery of each clause message needs an energy of at least
	$(1+j\varepsilon)(\tfrac{1-i\varepsilon}{1+i\varepsilon}+1) \geq 2$, with equality if and only if $j = i$. Summing over the clauses and variable
	boxes we get $\mathit{LB}' = 4xy + 2y + x(y^2+y+1)\varepsilon = \cost(\SAT)$.
\end{proof}

It remains to note that the first literal agent (of weight $1+y\varepsilon$) on the path chosen by the variable agent can walk over to the other path
and from there on to deliver a clause message. However such a schedule has higher energy consumption than $\cost(\SAT)$: the literal agent has to
cross two edges of length $\varepsilon$ and the required energy to do so has not been used in the estimated lower bound $\mathit{LB}' = \cost(\SAT)$. 
Hence we conclude:

\messageorderhardnessthm*
%

\section{Appendix: Approximation algorithm}
\label{app:approximation}

\approximationalgorithm*

\begin{proof}
	We start by artificially enlarging the weight $w_j$ of every agent $a_j$ to $\max w_i$. In doing so, we increase the energy cost contribution of each agent $a_j$ by a factor of $\tfrac{w_i}{w_j}$. 
	Thus any $\rho$-approximation to this weight-enhanced problem will give us a $\rho \cdot \max \tfrac{w_i}{w_j}$ approximation for the original problem.

	From now on assume without loss of generality that all agents have a uniform weight $w := \max w_i$. 
	Let $\OPT$ be an optimal schedule for an instance of \ourProblem with uniform agent weights $w$ and capacities $\kappa = 1$.
	We call a feasible schedule $S$ \emph{restricted} if in $S$ every message $m_i$ is transported by a single agent from $s_i$ to $t_i$ without any intermediate drop-offs. 
	By Theorem~\ref{theo-boc-nointermediate} there exists a restricted schedule $S^R$ with $\cost(S^R) \leq 2\cdot \cost(\OPT)$. 
	Let $\OPT^R$ be any optimal restricted schedule with total energy consumption $\cost(\OPT^R) \leq \cost(S^R) \leq 2\cdot \cost(\OPT)$.
	
	We define a complete undirected auxiliary graph $G' = (V',E')$ on all agent starting positions as well as all message sources and destinations, 
	$V' = \left\{ p_1, \ldots, p_k \right\} \cup \left\{ s_1, \ldots, s_m \right\} \cup \left\{ t_1, \ldots, t_m \right\}$. 
	Each edge $e=(u,v) \in E'$ has length $l_e := d_G(u,v)$.
	$\OPT^R$ has a natural correspondence to a vertex-disjoint path cover $PC(\OPT^R)$ of $G'$ with exactly $k$ simple paths $P_1, \ldots, P_k$ that has the following properties: 
	\begin{itemize}
		\item	Each path $P_j$ contains exactly one agent starting position, namely $p_j$;\\
			and $p_j$ is an endpoint of $P_j$.
		\item	Each destination node $t_i$ is adjacent to its source node $s_i$;\\
			and $s_i$ lies between $t_i$ and the endpoint $p_j$. 
	\end{itemize}

	Each path (possibly of length $0$) with endpoint at a starting position $p_j$ corresponds to the (possibly empty) schedule $\OPT^R|_{a_j}$ and $\cost(\OPT^R) = \sum_{e \in PC(\OPT^R)} l_e$.
	Now let $TC^*$ denote a tree cover of minimum total length $\sum_{e \in TC^*} l_e$ among all vertex-disjoint tree covers $TC$ of $G'$ with exactly $k$ trees $T_1,\ldots,T_k$ that satisfy the following properties:
	\begin{itemize}
		\item	Each tree $T_j$ contains exactly one agent starting position, namely $p_j$.
		\item	Each destination node $t_i$ is adjacent to its source node $s_i$.
	\end{itemize}

	Since $PC(\OPT^R)$ itself is a tree cover satisfying the two mentioned properties, we immediately get $\sum_{e \in TC^*} l_e \leq \sum_{e \in PC(\OPT^R)} l_e = \cost(\OPT^R)$. 
	By Theorem~\ref{thm:planning-restricted} we can use DFS-traversals in each component of the optimum tree cover $TC^*$ to construct in polynomial-time a schedule $S^*$
	of total energy consumption $\cost(S^*) \leq 2\cdot \sum_{e \in TC^*} l_e \leq 2\cdot \cost(\OPT^R) \leq 4\cdot \cost(\OPT)$.

	It remains to show that the tree cover $TC^*$ can be found in polynomial time: 
	Analogously to Theorem~\ref{thm:planning-restricted} we start with an empty graph on $V'$ to which we add all edges $\left\{ (s_i, t_i)\ | \ i \in \left\{ 1, \ldots, m \right\} \right\}$.
	Then we add all other edges of $E'$ in increasing order of their lengths, disregarding any edges which would result \emph{either} in the creation of a cycle 
	\emph{or} in a join of two starting positions $p_i, p_j$ into the same tree.
\end{proof}
\input{appendix-other.tex}
\fi

\end{document}